\documentclass[a4paper]{article}

\usepackage[T1]{fontenc}
\usepackage[utf8]{inputenc}
\usepackage{float}
\usepackage{graphicx}
\usepackage{subcaption} 
\usepackage[thinlines]{easytable}
\usepackage{hyperref}
\usepackage{wrapfig}
\usepackage{dingbat}
\usepackage{color}
\usepackage{refcount}
\usepackage[table]{xcolor}
\usepackage[toc,page]{appendix}
\usepackage{booktabs}

\usepackage{indentfirst}

\usepackage{amsmath,amssymb,amsthm}
\usepackage{mathtools}

\usepackage[english]{babel}

\usepackage{natbib}

\usepackage{xcolor}

\makeatletter
\@addtoreset{equation}{section}

\makeatother

\makeatletter
\newcommand{\douwidehat}[2]{%
  \sbox0{$\m@th#1\widehat{\hphantom{#2}}$}%
  \sbox2{$\m@th#1x$}
  \sbox4{$\m@th#1#2$}
  \dimen0=\ht0
  \advance\dimen0 -.8\ht2
  \dimen2=\dp4
  \rlap{%
    \raisebox{\dimexpr\dimen0-\dimen2}{%
      \scalebox{1}[-1]{\box0}%
    }%
  }%
  {#2}%
}
\makeatother

\theoremstyle{Theorem}
\newtheorem{theo}{Theorem}
\theoremstyle{Lemma}
\newtheorem{lem}[theo]{Lemma}
\newtheorem{rem}[theo]{Remark}
\newtheorem{assumption}[theo]{Assumptions}

\newtheorem{prop}{Proposition}

\theoremstyle{Definition}

\theoremstyle{Corollary}

\DeclareFontFamily{U}{matha}{\hyphenchar\font45}
\DeclareFontShape{U}{matha}{m}{n}{
      <5> <6> <7> <8> <9> <10> gen * matha
      <10.95> matha10 <12> <14.4> <17.28> <20.74> <24.88> matha12
      }{}
\DeclareSymbolFont{matha}{U}{matha}{m}{n}
\DeclareFontSubstitution{U}{matha}{m}{n}

\DeclareFontFamily{U}{mathx}{\hyphenchar\font45}
\DeclareFontShape{U}{mathx}{m}{n}{
      <5> <6> <7> <8> <9> <10>
      <10.95> <12> <14.4> <17.28> <20.74> <24.88>
      mathx10
      }{}
\DeclareSymbolFont{mathx}{U}{mathx}{m}{n}
\DeclareFontSubstitution{U}{mathx}{m}{n}

\DeclareMathDelimiter{\vvvert}{0}{matha}{"7E}{mathx}{"17}

\newcommand{\E}{{\mathbb E}}

\newcommand{\N}{{\mathbb N}}
\newcommand{\Z}{{\mathbb Z}}

\newcommand{\be}{\begin{equation}} 
\newcommand{\ee}{\end{equation}}

\numberwithin{theo}{section}

\begin{document}

\title{Inference for Forecasting Accuracy: Pooled versus Individual Estimators in High-dimensional Panel Data}

\author{
Tim Kutta\thanks{Department of Mathematics, Aarhus University, \texttt{tim.kutta@math.au.dk}; corresponding author} \and
Martin Schumann \thanks{School of Business and Economics, Maastricht University, \texttt{m.schumann@maastrichtuniversity.nl}
} \and
Holger Dette\thanks{Fakultät für Mathematik, Ruhr-Universität Bochum, \texttt{holger.dette@rub.de}
}
}

\maketitle

\begin{abstract}
Panels with large time $(T)$ and cross-sectional $(N)$ dimensions are a key data structure in social sciences and other fields. A central question in panel data analysis is whether to pool data across individuals or to estimate separate models. Pooled estimators typically have lower variance but may suffer from bias, creating a fundamental trade-off for optimal estimation.
We develop a new inference method to compare the forecasting performance of pooled and individual estimators. Specifically, we propose a confidence interval for the difference between their forecasting errors and establish its asymptotic validity. Our theory allows for complex temporal and cross-sectional dependence in the model errors and covers scenarios where $N$ can be much larger than $T$—including the independent case under the classical condition $N/T^2 \to 0$. The finite-sample properties of the proposed method are examined in an extensive simulation study.
\end{abstract}

\section{Introduction}
Panel data arise in settings where information is observed across both individuals and time. In many applications, panels feature a large number of individuals $N$ and a non-negligible time dimension $T$, enabling researchers to study both cross-sectional heterogeneity and temporal dynamics. Most practically used panel models account for individual heterogeneity through unit-specific intercepts, but the vast majority impose (complete) homogeneity on the slope coefficients of interest.
This modeling choice effectively implies the use of pooled estimators that aggregate data across individuals. While pooling can improve efficiency by reducing variance, it may lead to misleading conclusions when slope coefficients differ across units. Such heterogeneity is frequently observed in empirical work (see, for discussions, \citealp{hsiao2008random}, \citealp{baltagi2008pool}, \citealp{browning2007heterogeneity} among many others).

To make data-driven decisions about whether to pool or not, one strand of research—following the seminal work of \cite{Swamy}—has developed tests for cross-sectional slope homogeneity. Examples include \cite{phillips-sul}, \cite{Pesaran}, and \cite{ando2015}. These tests are typically derived under the assumption that model errors are independent across both individuals and time, with a few exceptions such as \cite{Blomquist}, who consider scenarios with small $N$ and $T \to \infty$.
While cross-sectional slope homogeneity is certainly sufficient to justify pooling, it is not a necessary condition. In large-$N$ panels, complete slope homogeneity is often unrealistic, and consequently a number of less restrictive criteria have been explored in the literature. One important example involves latent group structures, where slopes are homogeneous within but not across groups (see \citealp{su2016identifying}; \citealp{wang2018homogeneity}). A recent extension by \cite{wang2024homogeneity} considers models with a large number of covariates. A quantitative measure for slope homogeneity in large-$N$, large-$T$ settings was proposed by \cite{kutta:dette:2024}, where small values of this measure suggest the use of pooled estimators. Another, more direct approach to assessing the cost of pooling was introduced by \cite{Campello}, who developed tests for the hypothesis of no slope heterogeneity bias. While this hypothesis can still be practically restrictive for large $N$, this contribution is important by widening the discussion about what counts as a reasonable pooling criterion. 

In this paper, we investigate a different pooling criterion that is directly motivated by prediction—a primary concern for applied users. Rather than testing for precise slope homogeneity (possibly after grouping), we ask which estimators—individual or pooled—yield lower prediction errors. This approach is motivated by the insight that in most practical applications, some degree of slope heterogeneity is unavoidable, even within groups. The key question is therefore whether heterogeneity can be safely ignored or whether it meaningfully affects forecasting accuracy.
Our pooling criterion is based on the mean squared forecasting error (MSFE), first suggested for this purpose by \cite{pesaran:pick:timmermann:2022}. While we adopt minimal MSFE as a decision rule, the statistical model and asymptotic framework developed here differ substantially from those in \cite{pesaran:pick:timmermann:2022}. The most important conceptual distinction is that we allow for fixed (non-random) slopes, whereas \cite{pesaran:pick:timmermann:2022} derive their inference methods within the random coefficient model of \cite{Swamy}. The latter framework is particularly suitable when the goal is consistent estimation of a mean slope coefficient. However, as we argue in this paper, assuming that individual slopes—or their deviations from the mean—follow a specific distribution is both unnecessary and potentially restrictive. Moreover, \cite{pesaran:pick:timmermann:2022} develop, technically speaking, an oracle procedure for fixed $T$ and $N \to \infty$, though in practice a large $T$ seems necessary for variance estimation. Our approach, in contrast, is fully feasible and provably consistent for large time  and cross-sectional  dimension. We defer a detailed comparison to Remark \ref{rem:compa} below.
We summarize the main contributions of this paper as follows:
\begin{itemize}
\item[(1)] We develop the first asymptotic confidence interval for the difference between pooled and individual prediction errors in large data panels.
\item[(2)] Our methodology accommodates general error distributions and simultaneous temporal and cross-sectional dependence.
\item[(3)] We establish validity for large samples with $T, N \to \infty$, allowing for $N \gg T$. New asymptotics are derived under a “moderately heterogeneous” model, where the decision whether to pool is most challenging and statistical inference most valuable.
\end{itemize}

The remainder of this paper is organized as follows: Statistical methodology is developed and theoretically justified in Section \ref{section-2}.  Finite sample properties of our approach are studied in Section \ref{section-simulations}. Proofs and additional simulations are located in the Appendix.

\section{Inference for the prediction error in panel data}\label{section-2}

In this section, we develop new methodology to quantify the difference of prediction errors
between pooled and individual slope estimators. Methods are theoretically justified in an asymptotic
regime, where both the length of the time series $T$ and the size of the cross-section $N$ are large. To derive 
meaningful probabilistic limits, theory is developed for regimes of \textit{moderate
heterogeneity} (Assumption \ref{ass_3}), where the decision for the superior forecasting method is most challenging. The main statistical result is the Gaussian approximation for the distribution of the estimated forecasting error in Theorem \ref{theo_main}, which implies confidence intervals for the forecasting error \eqref{e:defCI}. We begin this section, by introducing some mathematical concepts which are important for our subsequent mathematical analysis.

\subsection{Mathematical preliminaries} \label{sec_prelim}

\noindent \textbf{Strongly mixing panels} In this work, we study panels of dependent random variables. Dependence is captured by $\alpha$-mixing, with weak dependence expressed by fast decaying mixing coefficients. In the following, we provide a general outline of $\alpha$-mixing for multivariate panels of random variables \cite[see][]{doukhan:1994}. For an overview of mixing concepts, we refer the reader to \cite{bradley:2007}.\\
Let $\mathcal{M}\subset \Z^d$ for some $d \in \N$ be endowed with the maximum norm $\|\cdot \|_\infty$ and
 $(\Xi_{z})_{z\in \mathcal{M}  }$ be a panel of random variables indexed in $\mathcal{M}$.
 For index sets $\mathcal{I}$ and $\mathcal{J}$ we define the distance w.r.t. the maximum norm as
 $$
 \mbox{dist}(\mathcal{I}, \mathcal{J}) = \min(\|i-j\|_\infty: i \in \mathcal{I}, j \in \mathcal{J})
 $$
 and denote by $\mathcal{F}_\mathcal{I} =\sigma\big(\Xi_{z}: z \in \mathcal{I}\big)$ the sigma field generated by the random variables $\{ \Xi_{z}: z \in \mathcal{I}\} $.
For $r \in \N_0$ the $r$th $\alpha$-mixing coefficient  is
defined by 
\begin{align*}
\alpha(r)=\sup &\Big\{ |\mathbb{P}(A \cap B)- \mathbb{P}(A) \mathbb{P}(B)|:  A \in \mathcal{F}_\mathcal{I}, ~
B \in \mathcal{F}_\mathcal{J},~
dist(\mathcal{I}, \mathcal{J})\ge r,~ \mathcal{I}, \mathcal{J} \subset \mathcal{M}\Big\}.
\end{align*}
The  panel is called $\alpha$-mixing if $\alpha(r) \to 0$ as  $r \to \infty$. 
\noindent In the case of this work, we will consider the dimension $d=2$, even though extensions to higher dimensional panels are  possible.
\smallskip

\textbf{Conditional Landau symbols} One characteristic of this work is the development of conditional inference methods for panel data. We therefore have to clarify notions of conditional convergence of random variables. 
Suppose that $(X_n)_n, (Y_n)_n$  are sequences of random variables and $(a_n)_n$ is a sequence of positive, real numbers. Then we say, that
\begin{equation} \label{e:stochLan}
\begin{split}
X_n =\, &\mathcal{O}_P^{|Y_n}(a_n) \,\,  \Leftrightarrow \,\, \forall \delta>0: \lim_{C \to \infty}\limsup _n \mathbb{P}\Big (\mathbb{P} \big (|X_n|/a_n > C~|~Y_n \big  )>\delta\Big ) =0,\\
X_n =\, & o_P^{|Y_n}(a_n) \,\,\,\,  \Leftrightarrow \,\, \forall C, \delta>0:\,\,\lim _n \mathbb{P}\Big (\mathbb{P} \big  (|X_n|/a_n > C~|~ Y_n \big  )>\delta\Big ) =0.
\end{split}
\end{equation}

\subsection{Statistical methodology}

 \noindent \textbf{Panel model} We  consider the linear regression panel 
\begin{equation} \label{model_1}
y_{i,t} =x_{i,t}' \beta_i+\varepsilon_{i,t} \quad \quad i=1,...,N, \,\, t=1,...,T,
\end{equation}
where $x_{i,t}$ is a $K$-dimensional vector of regressors, $\beta_i$ is a $K$-dimensional vector of slope coefficients and $\varepsilon_{i,t}$ a centered, real model error with unknown distribution.
We call $t$ the \textit{time component} of the panel and $i$ the \textit{individual} or \textit{cross-sectional component}.  We  collect all equations concerning the $i$th individual in the model
\begin{equation} \label{model_2}
y_{i} =X_{i} \beta_i+\varepsilon_{i} \quad \quad  i=1,...,N,
\end{equation}
where $X_{i}=(x_{i,1},...,x_{i,T})' $ is a regression matrix of dimension $T \times K$ and $y_{i}=(y_{i,1},...,y_{i,T})'$ and $\varepsilon_{i}=(\varepsilon_{i,1},...,\varepsilon_{i,T})'$ are  $T-$dimensional vectors. Sometimes we  refer to all regressors collectively and therefore define the compounded matrix $\mathbf{X}=(X_1,...,X_N)$.  
Often panel models comprise constant, individual specific intercepts, which are omitted in model \eqref{model_2} for simplicity, and would practically be removed. We illustrate this case in our simulation study in Section \ref{section-simulations}.
\smallskip

\textbf{Estimators} For model \eqref{model_1}, we define the individual ordinary least squares (OLS) slope estimator for the $i$th individual as
\[
\hat \beta_i := [X_i' X_i]^{-1} X_i' y_i.
\]
Next, drawing on data from all $N$ individuals, we define the pooled version 
\[
\hat \beta^{pool}:= \Big( \sum_{i=1}^N X_i' X_i \Big)^{-1} \sum_{i=1}^N X_i' y_i.
\]
The estimator $\hat \beta^{pool}$ is commonly used under the assumption of slope homogeneity $\beta_1=...=\beta_N$, where it is substantially more efficient than the individual estimators $\hat \beta_1,...,\hat \beta_N$ due to its smaller variance. However, the effectiveness of the pooled estimator relies on the individual slopes being, if not the same, at least very similar - otherwise, it can be severely biased. A standard way to assess slope homogeneity is the use of slope homogeneity tests  as discussed in the  Introduction. While these tests are very powerful for larger $N$, their very power often makes them oversensitive to even minor inhomogeneities of slopes, discouraging pooling even when practically  beneficial. In the below discussion, we therefore try to shift the subject, away from somewhat stylized assumptions on the true slopes, towards the comparative merits of the estimators in terms of forecasting.
\smallskip

\textbf{Individual and pooled MSE} The performance of the individual estimators $\hat \beta_1,...,\hat \beta_N$ and the pooled version $\hat \beta^{pool}$ can be assessed by comparing their prediction accuracy for a new vector of predictor-response pairs $(x_{i,T+1}, y_{i,T+1})_{i=1,...,N}$. Thus, following \cite{pesaran:pick:timmermann:2022}, we invoke the mean squared prediction error (MSPE) in $(x_{i,T+1},y_{i,T+1})_{i=1,...,N}$ conditionally on the known matrix of regressors $\mathbf{X}=(X_1,...,X_N)$. More precisely, we define respectively the individual and pooled prediction errors as
\begin{align}\label{E_ind}
    E^{ind}:= & \frac{1}{N}\sum_{i=1}^N E_i^{ind} := \frac{1}{N}\sum_{i=1}^N\E \big[ ( x_{i,T+1}'\hat \beta_i -y_{i,T+1})^2\big| \mathbf{X}\big],\\
     E^{pool}:=& \frac{1}{N}\sum_{i=1}^N E_i^{pool} := \frac{1}{N}\sum_{i=1}^N\E \big[ (x_{i,T+1}'\hat{\beta}^{pool} -y_{i,T+1})^2\big| \mathbf{X}\big]. \label{E_pool}
\end{align}
Notice that in this formulation, the vector of predictors $(x_{i,T+1})_{i=1,...,N}$ is non-random - the user determines in which predictors they would like to make the comparison. Selecting $(x_{i,T+1})_{i=1,...,N}$ means specifying a scenario where forecasts are of interest and we seek to determine which forecasting method is most suitable for it.
The responses at time $T+1$ are defined as $y_{i,T+1}:=x_{i,T+1}'\beta_i+\varepsilon_{i,T+1}$ with model errors $\varepsilon_{i,T+1}$ and the expectations in  $E^{ind}$ and $E^{pool}$ are taken over all model errors $\varepsilon_{i,t}$, $i=1,...,N$ and $t=1,...,T+1$.
\smallskip

\textbf{A closed form for the prediction error} For our formal analysis of the prediction error, we impose some mathematical assumptions. In the following, let $\|\cdot\|_2$ denote the Euclidean norm (Frobenius norm) for vectors and matrices. 
\begin{assumption} \label{ass_1}
\begin{itemize} $ $
\item[i)] (Moments) For some $M \in \mathbb{N}$ with $M\ge 16$ and a constant $C>0$, the errors $\varepsilon_{i,t}$ and regressors $x_{i,t}$ satisfy
    \[
    \max_{i,t}\mathbb{E}|\varepsilon_{i,t}|^M \le C<\infty, \qquad \max_{i,t}\mathbb{E}\|x_{i,t}\|_2^M \le C<\infty.
    \]
    \item[ii)] (Error covariance)  The errors $\varepsilon_{i,t}$ are centered. The vector of errors  $\varepsilon = (\varepsilon_1',....,\varepsilon_N')' \in \mathbb{R}^{NT}$  has 
    a covariance matrix $\Sigma = \Sigma_N \otimes \Sigma_T$, where "$\otimes$" denotes the Kronecker product, $\Sigma_T \in \mathbb{R}^{T \times T}$ is a  covariance matrix of a stationary process and $\Sigma_N \in \mathbb{R}^{N \times N}$ is another covariance matrix.
    \item[iii)] (Exogeneity) The vector of errors $\varepsilon$  and the matrix of regressors $\mathbf{X}$ are independent of each other.
    \item[iv)] (Prediction) The regressors $x_{i,T+1}$ are fixed, non-random vectors. The vector  $(\varepsilon_{1,T+1},...,\varepsilon_{N,T+1})$ is independent of $(\varepsilon, \mathbf{X})$, centered and has covariance matrix $(\Sigma_T)_{1,1}\cdot \Sigma_N$. 
\end{itemize}
\end{assumption}
\noindent Conditions $i)$ and $ii)$ 
permit the existence of complex error structures. Error distributions are non-parametric and only some polynomial moments are required. Moreover, the errors can be dependent across space and time simultaneously.
 More precisely, the covariance matrix of the  errors $(\varepsilon_1^\prime, \ldots ,\varepsilon_N^\prime)^\prime$  is the Kronecker product $\Sigma_N \otimes \Sigma_T$, where the first factor captures cross-sectional and the second factor temporal dependence.  This structure is more general than those typically assumed 
in the literature and which  are special cases of this setting.  For example, the traditional assumptions of independent, homoscedastic errors is captured by $\sigma^2 \cdot I_N \otimes I_T$, where $\sigma^2>0$ is the variance and  $I_N, I_T$ are the identity matrices of dimension $N$ and $T$, respectively. Independent errors with individual-specific variances are similarly captured by $D \otimes I_T$, where $D=diag(\sigma_1^2,...,\sigma_N^2)$. The case of cross-sectional dependence only, as used e.g. by \cite{Zellner}, is incorporated by $\Sigma_N \otimes I_T$  (this is also closely related to \cite{ando2015}). Allowing both factors to differ from the identity matrix (simultaneous temporal and cross-sectional dependence) obviously encompasses much richer models than have been treated before, particularly in high dimensional panels. Separable covariance structures are a standard tool in the analysis of spatio-temporal data and often appropriate for panels, where the individual component can, in a broad sense, be interpreted as a location. Condition $iii)$ is standard in the study of linear models and stronger than the common exogeneity assumptions used in the literature on panel data. We require it, to rigorously formulate weak convergence results conditionally on the regressors $\mathbf{X}$, even though we expect similar results to be true if errors and regressors are weakly dependent. Finally, Condition $iv)$ implies that the errors in the (hypothetical) time period $T+1$ are independent of all previous errors and maintain the same covariance structure.    \\
Using Assumption \ref{ass_1}, we can give a closed form for the prediction errors defined in \eqref{E_ind} and \eqref{E_pool}.
\begin{lem} \label{lem_1} Suppose that Assumption \ref{ass_1} holds, then
\begin{align*}
      E_i^{ind} = &  (\Sigma_N)_{i,i} Tr \Big[ \Sigma_T \Big(X_i  [X_i' X_i]^{-1}x_{i,T+1}x_{i,T+1}'[X_i' X_i]^{-1}X_i' \Big)\Big]+(\Sigma_N)_{i,i}(\Sigma_T)_{1,1}. \\
         E_i^{pool} = &\bigg(x_{i,T+1}'\Big( \sum_{j=1}^N X_j' X_j \Big)^{-1} \sum_{j =1}^N X_j' X_j (\beta_j-\beta_i)\bigg)^2\\
   & +Tr \bigg[ \Big( \sum_{j=1}^N X_j' X_j \Big)^{-1} x_{i,T+1} x_{i,T+1}'\Big( \sum_{j=1}^N X_j' X_j \Big)^{-1}\sum_{j,k}(\Sigma_N)_{j,k} X_j' \Sigma_T X_k\bigg]+(\Sigma_N)_{i,i} (\Sigma_T)_{1,1}.
\end{align*}
    and in particular $E^{ind}-E^{pool} = E_1 -E_2 -E_3$, where
    \begin{align}
   E_1&:= 
   Tr \Big[ \Sigma_T \sum_{i=1}^N\Big(\frac{1}{N}(\Sigma_N)_{i,i} X_i  [X_i' X_i]^{-1}x_{i,T+1}x_{i,T+1}'[X_i' X_i]^{-1}X_i' \Big)\Big] \label{e:E_1}\\
    E_2&:=\frac{1}{N}\sum_{i=1}^N\bigg(x_{i,T+1}'\Big( \sum_{j=1}^N X_j' X_j \Big)^{-1} \sum_{j =1}^N X_j' X_j (\beta_j-\beta_i)\bigg)^2 \label{e:E_2} \\
    E_3& :=Tr \bigg[ \Big( \sum_{j=1}^N X_j' X_j \Big)^{-1}\frac{1}{N}\sum_{i=1}^N \big\{x_{i,T+1} x_{i,T+1}'\big\}\Big( \sum_{j=1}^N X_j' X_j \Big)^{-1}\sum_{j,k}(\Sigma_N)_{j,k} X_j' \Sigma_T X_k\bigg]\label{e:E_3}
    .
\end{align}
\end{lem}
\noindent Our main object of interest is the difference $E^{ind}-E^{pool}$, because, e.g., $E^{ind}-E^{pool}<0$ implies that the pooled estimator is outperformed by the individual estimators. Lemma \ref{lem_1} now implies that to understand $E^{ind}-E^{pool}$ we may study the terms $E_1,E_2,E_3$.  \\
\smallskip

\textbf{Analysis of $E_i$} We first have to impose some additional assumptions. 
\begin{assumption}$ $ \label{ass_2}
    \begin{itemize}
        \item[i)] (Regressor convergence) For each $i=1,...,N$ there exists a positive definite matrix $Q_i$, such that 
        \begin{align*}
           \max_i \|X_i'X_i/T-Q_i\|_2 \overset{\mathbb{P}}{\to} 0, \quad  0< c_1 \le  \min_i \lambda_{min}(Q_i), \quad \max_i \|Q_i\|_2  \le c_2.
    \end{align*}
    Here, $\lambda_{min}$ denotes the smallest eigenvalue of a matrix and  $c_1, c_2$ some fixed positive constants.
    \item[ii)] (Boundedness $x_{i,T+1}$) There exists a constant $c_3>0$, such that $\max_i \|x_{i,T+1}\|_2 \le c_3$.
    \item[iii)] (Errors) The errors $\{\varepsilon_{i,t}: 1 \le i \le N, 1 \le t \le T\}$ form a strongly mixing field, with mixing coefficients satisfying for all $r$ sufficiently large, and some constant $\psi>0$ that
    $\alpha(r) \le \psi^{-r}$.
   \end{itemize}
\end{assumption}

\noindent Assumptions of the sort imposed by Condition $i)$ are typical in the study of large panel data (see, e.g., Assumption 2 in \cite{Pesaran}). The regressor matrices converge to respective limits $Q_i$ that are positive definite (smallest eigenvalue is positive) and have uniformly bounded norm. The second condition  ensures that the forecasting error is calculated for $x_{i,T+1}$ that are not too large, which can be seen as a condition to avoid extrapolating into areas that are far apart from the original data. Finally, we assume that the model errors are weakly dependent along the space and time dimensions. This structure allows for general dependence patterns that are more realistic than traditional models with independent errors. The exponential decay condition on mixing coefficients is satisfied by most of the typical time series models such as ARMA processes  \cite[see][]{Mo88}. We also point out that in the Appendix we prove our results for even weaker, polynomial mixing conditions (see condition \eqref{e:mix:app}). Such conditions come at the cost of a more restrictive relation between $N$ and $T$ (see condition \eqref{e:eta:app}) and are therefore not further discussed here.
\smallskip

\textbf{Analysis of the variance terms}  As a first step to analyze the error terms $E_1, E_2, E_3$ from Lemma \ref{lem_1}, we investigate their (asymptotic) order of magnitudes. We begin by studying the two terms $E_1$ and $E_3$ that are independent of the regression slopes and represent the variance of the individual estimators and the pooled estimator, respectively. Convergence is formulated conditionally on $\mathbf{X}$ to harmonize this  with later results, but obviously, $E_1, E_3$ are $\mathbf{X}$-measurable and hence the derived rates hold conditionally and unconditionally.

\begin{lem} \label{lem_2}
    Suppose that Assumptions \ref{ass_1} and \ref{ass_2} hold, then
    \[
    E_1 = \mathcal{O}^{|\mathbf{X}}_P(T^{-1}), \qquad \textnormal{and} \qquad E_3 = \mathcal{O}_P^{|\mathbf{X}}((NT)^{-1}).
    \]
\end{lem}

\noindent  The derived rates of convergence are highly intuitive. For the individual slope estimators that are based (each) on $T$ observations, the variance of predictions are of size $\approx C/T$. Similarly, the predictions based on the pooled estimator have variance of size $\approx C/(NT)$, which is much smaller. While the variance of the pooled predictions are much smaller, the pooled estimator can produce biased results if the slopes are different. Here, the bias is measured by $E_2$. Under the classical hypothesis of slope homogeneity ($\beta_1= \ldots  =\beta_N$) it directly follows that $E_2=0$. In other, very heterogeneous scenarios, $E_2=\mathcal{O}_P(1)$ is possible such that $E_2$ dominates both $E_1$ and $E_3$. These two extremes illustrate situations, where a pooled estimator is either evidently superior (because of high homogeneity) or evidently inferior  (because of high heterogeneity) compared to individual estimators. In applications, the situation is often less clear-cut. For this reason, we investigate in the following the practically more relevant intermediate case,  where $E_2$ is of the same order of magnitude as $E_1$. In this scenario, statistical inference can be invoked to make an informed decision on which prediction method is better. We hence impose further assumptions.

\begin{assumption} \label{ass_3}
\begin{itemize} $ $
    \item[i)] (Bounded slopes) The slopes $\beta_1, \beta_2,...$ are uniformly bounded  $\max_i\|\beta_i\|\le c_4$ for some $c_4>0$. 
    \item[ii)] (Moderate heterogeneity) $\max_{i,j} \|\beta_i-\beta_j\| \le c_5/\sqrt{T}$ for some $c_5>0$.
    \item[iii)] (Intersect-time relation) There exists a constant $\eta \in (0,2)$ such that $N/T^\eta\to 0$. 
\end{itemize}
\end{assumption}

\noindent The first condition is standard in the literature on slope homogeneity and guarantees that no single slope dominates the final test statistics. The second condition helps us focus on the main case of interest, where the two errors, individual and pooled, are of the same order of magnitude. Lemma \ref{lem_2} entails that $E_3$ is negligible compared to $E_1$ and hence, for $E_1-E_2-E_3$  to be close to $0$, $E_2$ has to be of the same order as $E_1$. We provide a short calculation to illustrate that under moderate slope heterogeneity $E_2=\mathcal{O}_P(T^{-1})$ (just as $E_1$).
Using Assumption \ref{ass_2} part i), we have
\begin{align*}
   E_2= & \frac{1}{NT}\sum_{i=1}^N\bigg(x_{i,T+1}'\Big[ \frac{1}{NT}\sum_{j=1}^N X_j' X_j \Big]^{-1} \frac{1}{N}\sum_{j =1}^N \frac{X_j' X_j }{T}[\sqrt{T}(\beta_j-\beta_i)] \bigg)^2 \\
   \approx & \frac{1}{T}\bigg\{\frac{1}{N}\sum_{i=1}^N\bigg(x_{i,T+1}'\Big[ \frac{1}{N}\sum_{j=1}^N Q_j\Big]^{-1} \frac{1}{N}\sum_{j =1}^N Q_j[\sqrt{T}(\beta_j-\beta_i)] \bigg)^2\bigg\}.
\end{align*}
Now, consider the right side: Under Assumption \ref{ass_3} part ii), the object on the inside of the round bracket is of order $\mathcal{O}(1)$ and hence the entire term inside the curved bracket is also of size $\mathcal{O}(1)$, suggesting that $E_2 = \mathcal{O}_P(T^{-1})$ (see Lemma \ref{lem_2}). 
It is clear that the formulation of Condition $ii)$ can be relaxed such that not all slopes have to be close to one another, but only the majority. For instance, one may impose that 
 there exists an index set $\mathcal{N}\subset \{1,...,N\}$ of exceptions such that only the weaker assumption
 \[
 \max_{i,j \in \{1,...,N\}\setminus \mathcal{N}}\|\beta_i-\beta_j\|\le c_5/\sqrt{T}
 \]
 holds. As long as $\mathcal{N}$ is small enough, say if $|\mathcal{N}|=o(\sqrt{N})$, the results in this section remain valid.
 The final Condition iii) in Assumption \ref{ass_3} moderates the size of the temporal dimension of the panel, compared to its cross-sectional component. The assumption $N/T^2 \to 0$ has been used in classical slope homogeneity tests \cite[see the discussion in][] {Pesaran} and the fact that in our case $N/T^\eta \to 0$ for $\eta$ arbitrarily close to $2$ is required is a small additional price that we pay for allowing temporal and cross-sectional dependence.  Finally, notice that Assumption \ref{ass_3} does not impose any distribution on the slope coefficients as is typically the case in random coefficient models. While distributional assumptions are convenient when the focus is on consistent estimation of the average slope coefficient, these assumptions can prove to be restrictive when the main objective is prediction.
 \smallskip
 
 \textbf{Statistical inference} We now develop an inference method for the difference of prediction errors. As a first step, we propose the estimator
\begin{align}\label{e:Ehat}
    \hat E := & \frac{1}{N}\sum_{i=1}^N\bigg(x_{i,T+1}'\Big( \sum_{j=1}^N X_j' X_j \Big)^{-1} \sum_{j =1}^N X_j' X_j (\hat \beta_j-\hat \beta_i)\bigg)^2. 
\end{align}
On the first glance, $\hat E$ may seem like a simple plug-in estimator for $E_2$ (see eq. \eqref{e:E_2}). Yet, a careful analysis reveals that  $\hat E$ actually approximates the error sum $E_1+E_2$. It can therefore serve as a building block to approximate our true object of interest $E_1-E_2$. 

\begin{lem} \label{lem_3}
    Suppose that Assumptions \ref{ass_1} and \ref{ass_2} hold. Then
    \[
    |\hat E-(E_1+E_2)|= \mathcal{O}^{|\mathbf{X}}_P(T^{-1}N^{-1/2}).\qquad 
    \]
\end{lem}

\noindent This lemma is a consequence of the much more precise analysis in the proof of Theorem \ref{theo_main}, where we study the weak convergence behavior of an appropriately standardized version of $\hat E$. We state the lemma here, to make the next step of our procedure more understandable. We recall that for our inference method, we do not need an estimator for $(E_1+E_2)$, but rather for $(E_1-E_2)$. Therefore, we supplement $\hat E$ with additional estimator  $\hat E_1$ of $E_1$. We will then have
\[
\hat E-2\hat E_1\approx E_2-E_1.
\]
Let us define the $T \times T$ matrix
\begin{equation} \label{def_hat_Sigma_i}
\hat \Sigma^{(i,k)}(b) := (\hat \xi^{(i,k)}(|s-t|) \mathbb{I}\{|s-t| < b\})_{1\le s,t\le T },
\end{equation}
where  the entry
\begin{eqnarray} \label{def_hat_xi}
\hat \xi^{(i,k)}(h) &:=& \frac{(y_i-X_i\hat \beta_i )_{1:T-h}' (y_k-X_k\hat \beta_k)_{h+1:T}}{T -h-K} \\
\nonumber &=& \frac{1}{T -h-K} \sum_{t=1}^{T -h} (y_i-X_i\hat \beta_i )_t (y_k-X_k\hat \beta_k )_{t+h}
\end{eqnarray}
is  an  estimator  of the autocovariance of lag $h$.
Here  $b$ is a regularization  parameter (all $h$-diagonals with $h>b$ are set equal to $0$). Finally, we define the estimate
\begin{equation} \label{e:hE1}
 \hat E_1:= \frac{1}{N}\sum_{i=1}^NTr\Big[
    \hat \Sigma^{(i,i)}( b) X_i  [X_i' X_i]^{-1}x_{i,T+1}x_{i,T+1}'[X_i' X_i]^{-1}X_i' \Big ] =: \frac{1}{N}\sum_{i=1}^N \hat E_1^{(i)}.
\end{equation}
Notice that $ \hat E_1$ is the direct, empirical analogue to $E_1$ defined in \eqref{e:E_1}.
\begin{lem}\label{lem_4}
    Under Assumptions \ref{ass_1}, \ref{ass_2} and \ref{ass_3} it holds that 
    \[
    |\hat E_1-E_1|=\mathcal{O}^{|\mathbf{X}}_P\Big(\frac{b}{T^2}+\frac{ b^{7/4}}{\sqrt{N}T^{3/2}}\Big).
    \]
    If we choose $b=T^\rho$ with 
    \[
    0 <\rho < \min\Big(1-\frac{\eta}{2},\frac{2}{7} \Big)
    \]
    the right side is of order $o^{|\mathbf{X}}_P(1/(\sqrt{N}T))$.
\end{lem}

\noindent  Notice that the above Lemmata now imply that 
\[
\sqrt{N}T\big\{(\hat E-2\hat E_1)-(E^{ind}-E^{pool})\big\} = \sqrt{N}T \big\{\hat E-(E_1-E_2)\big\}+o^{|\mathbf{X}}_P(1).
\]
We can thus develop statistical inference for the difference $E^{ind}-E^{pool}$, by studying the weak convergence of the statistic $\sqrt{N}T \{\hat E-(E_1-E_2)\}$.
 For this purpose, we define the following two terms
\begin{align*}
    \Lambda := & \Big( \sum_{j=1}^N \frac{X_j' X_j}{NT} \Big)^{-1} \Big\{ \sum_{i=1}^N \frac{x_{i,T+1}x_{i,T+1}'}{N}\Big( \sum_{j=1}^N \frac{X_j' X_j}{NT} \Big)^{-1} \sum_{j =1}^N \frac{X_j' X_j (\beta_j-\beta_i)}{N\sqrt{T}}\Big\},\\
    \Lambda_k := & \Big(\frac{X_k'X_k}{T}\Big)^{-1} x_{k,T+1}x_{k,T+1}'\Big( \sum_{j=1}^N \frac{X_j' X_j}{NT} \Big)^{-1} \sum_{j =1}^N \frac{X_j' X_j (\beta_j-\beta_k)}{N \sqrt{T}}.
\end{align*}
Therewith, we define the conditional asymptotic variance
\begin{align} \label{e:deftau}
     \tau_N^2 = & \frac{1}{N}\sum_{i,k=1}^N \bigg\{2(\Sigma_N)_{i,k}^2 (x_{i,T+1}'(X_i'X_i/T)^{-1}[X_i'\Sigma_T X_k'/T](X_k'X_k/T)^{-1}x_{k,T+1})^2 \\
    &\qquad \qquad \quad+4 ( \Lambda+ \Lambda_i)' \frac{X_i' (\Sigma_N)_{i,k} \Sigma_T X_k}{T}  ( \Lambda+ \Lambda_k)\bigg\} \nonumber
\end{align}
As we show in the proof of Theorem \ref{theo_main}, $\tau_N^2$ is asymptotically close to the variance of $\hat E$ and can thus be used for standardization. As common in the study of dependent time series, we have to ensure that the variance does not asymptotically degenerate. 
\begin{assumption} \label{ass_var}$ $
    \begin{itemize}
        \item[] (Non-degenerate variance) The random variable $\tau_N$ satisfies the two conditions
        \[
        \tau_N = \mathcal{O}_P(1), \qquad \tau_N^{-1} = \mathcal{O}_P(1). 
        \]
    \end{itemize}
\end{assumption}
\noindent Let us define the standardized estimator
\begin{align} \label{e:ebreve}
\breve E:= \frac{\sqrt{N}T}{\tau_N}\big\{(\hat E-2\hat E_1)-(E^{pool}-E^{ind})\big\} \sim F_{\breve E}, 
\end{align}
where $F_{\breve E}$ refers to the cumulative distribution function of $\breve E$, conditional on $\mathbf{X}$ and $\tau_N$ is defined in \eqref{e:deftau}.

\begin{theo} \label{theo_main}
    Suppose that Assumptions \ref{ass_1}, \ref{ass_2}, \ref{ass_3} and \ref{ass_var} hold. Denoting by $\Phi $ the cumulative distribution function of a standard normal distribution, we   have 
    \[
    \mathbb{E} \Big [ \sup_{z \in \mathbb{R}} |F_{\breve E}(z)-\Phi(z)| \big ] =o(1).
    \]
\end{theo}

\noindent The proof of Theorem \ref{theo_main} is technically challenging in two respects. First, the derivation of conditional weak convergence requires a non-standard application of a Berry-Esseen theorem for random variables conditional on the regressors. In the scenario of moderate heterogeneity an unusual linearization of $\breve E$ occurs that features both linear and squared error terms. These different terms influence the variance $\tau_N^2$, where two terms in the curly brackets occur, one for the squared errors (first term) and one for the non-squared errors (last term). Second, due to the complex dependence structure, it is difficult to prove that standardizing by $\tau_N^2$ is asymptotically equal to standardizing by the true variance of $\hat E$. 
In the special case of Gaussian data, the proof turns out to be much simpler because the squared and linear errors are uncorrelated; if $Z_1, Z_2$ are jointly normally distributed, $\mathbb{E}[Z_1Z_2^2]=0$ regardless of the covariance structure. Yet, for general error distribution, the analysis is substantially more intricate.
\smallskip

\textbf{A conditional confidence interval} Theorem \ref{theo_main} can be used to construct confidence intervals for the difference of prediction errors $E^{pool}-E^{ind}$. Suppose that the conditional variance $\tau_N^2$ is known. Then, an asymptotic $(1-\alpha)$ confidence interval for $E^{pool}-E^{ind}$ is given by 
\begin{equation}\label{e:defCI}
\mathcal{C}_{1-\alpha}:= \Big[ \frac{\tau_N\Phi^{-1}(\alpha)}{\sqrt{N}T} +\hat E-2\hat E_1, \frac{\tau_N\Phi^{-1}(1-\alpha)}{\sqrt{N}T}+\hat E-2\hat E_1\Big].
\end{equation}
Notice that $\mathcal{C}_{1-\alpha}$ is an approximate $1-\alpha$ confidence interval in the sense that
\begin{align} \label{e:conf:guar}
\mathbb{P}\big((E^{pool}-E^{ind})\in \mathcal{C}_{1-\alpha}|\mathbf{X}\big) \overset{\mathbb{P}}{\to} 1-\alpha.
\end{align}

\begin{rem}{\rm (Variance estimation)
    To  conduct statistical inference it is necessary to estimate the long-run variance $\tau_N^2$. In practice, an estimator might use additional prior knowledge available to the user, such as independence of the individuals across $i$ or the like. Yet, it is possible to construct fully non-parametric estimators for $\tau_N^2$. For this purpose, define $\hat \Lambda, \hat \Lambda_k$ as the sample versions of $\Lambda, \Lambda_k$, where the slopes $\beta_i$ have been replaced by the OLS estimators $\hat \beta_i$. Moreover, recall the covariance estimator $\hat \Sigma^{(i,k)}$. Finally, let $\mathcal{K}$ be a continuous, symmetric, non-negative function, such that $\mathcal{K}(x)$ is non-increasing for $x\ge 0$. We assume that $\mathcal{K}(0)=1$ and $\mathcal{K}(1)=0$  and let $b'>0$ be a bandwidth choice. Then, we define
    \begin{align}\label{e:def-tauhat}
        \hat \tau_N^2:= & \frac{1}{N }\sum_{i,k=1}^N \mathcal{K}\bigg(\frac{|i-k|}{b'}\bigg)\bigg\{2 \Big( (x_{i,T+1}'(X_i'X_i/T)^{-1}\big(X_{i}'\hat \Sigma^{(i,k)}(b) X_k)_{s,t}/T\big)(X_k'X_k/T)^{-1}x_{k^,T+1}) \Big)^2 \notag\\
    &\qquad \qquad \quad+4 (\hat  \Lambda+ \hat \Lambda_i)' \frac{X_i' \hat \Sigma^{(i,k)}(b) X_k}{T}  (\hat  \Lambda+ \hat \Lambda_k)\bigg\}.
    \end{align}
    Two bandwidth parameters have to be chosen in this estimator, namely $b, b'$, which quantify dependence across time and individuals respectively. Stronger dependence requires larger values of the bandwidth choices. In the case of temporal independence, one can simply choose $b=1$ and in the case of spatial independence $b'=1$.}   
\end{rem}

\begin{rem}\label{rem:compa}{\rm (Relation to existing work) We discuss certain aspects of the  contribution, compared to \cite{pesaran:pick:timmermann:2022}. We have already mentioned that work in our Introduction, where we have pointed out its pivotal role of defining the forecasting error as a pooling criterion. We now focus on more technical aspects of that work.
\begin{itemize}
    \item[i)] \cite{pesaran:pick:timmermann:2022} developed  their procedures in a different model from ours: slopes are random and a procedure is developed for fixed $T$ and as $N \to \infty$. The proposed approach is technically speaking an oracle procedure and $T \to \infty$ is needed to estimate variances for practical inference. In our paper, by contrast, we have modeled slopes as fixed. The reason is that a random model is not helpful when assessing the forecasting error for large $N,T$ scenarios. The forecasting error can be estimated at a precision of $1/\sqrt{NT}$, while the average of $N$ slopes fluctuates at a larger order of magnitude $1/\sqrt{N}$. So, with random slopes and large $N,T$, one either has to condition on the slopes (rendering them again fixed) or one has to see an inference procedure dominated not by the model errors, but by the "random variation of the slopes" which may be undesirable. For this reason, in our large $N,T$ panel setup, we believe that fixed slopes are the preferable model. We also notice that fixing slopes avoids imposing assumptions on the joint slope distribution, which may be restrictive in practice.
    \item[ii)] The inference tools developed in this paper and that of \cite{pesaran:pick:timmermann:2022} are different - we develop a confidence interval for the difference of forecasting errors, while they develop a hypothesis test. It turns out that these results are not equivalent, and we have made a confidence interval specifically to avoid challenges associated with conditional hypothesis testing. To be precise, \cite{pesaran:pick:timmermann:2022} develop a test for the point hypothesis of equivalent forecasting errors (eq. (19) of that work). This hypothesis is strictly speaking random and under standard assumptions only holds with probability $0$. Even when interpreted in a one-sided sense, it seems difficult to formulate standard asymptotic criteria  for statistical tests, such as an asymptotic error rate or consistency. In contrast, we develop an asymptotic confidence interval that, in a clearer way, holds the asymptotic level $1-\alpha$ (see eq. \eqref{e:conf:guar} above).
    \item[iii)] On a more basic level, theory in  \cite{pesaran:pick:timmermann:2022}  is developed for iid Gaussian model errors, while our approach can accommodate temporal and spatial dependence, and only imposes some weak moment assumptions on the error distributions.
\end{itemize}
}
\end{rem}

\section{Finite sample properties}\label{section-simulations}
We investigate the finite sample properties of our new confidence interval (CI) by means of a Monte Carlo study based on 5,000 iterations, with sample sizes $N\in\{100,500\}$ and $T\in\{10,15, 20,25,30, 40,60,80\}$. In each setup, we simulate $x_{i,t}$ as a $5$-dimensional vector, where each element is independently drawn from $\mathcal{N}(1,1)$. We report the empirical coverage rates of the feasible confidence interval $\hat{\mathcal{C}}_{1-\alpha}$, where $\tau_N$ in \eqref{e:defCI} is replaced by the square root of $\hat{\tau}^2_N$ defined in \eqref{e:def-tauhat}. As an infeasible benchmark, denoted as $\mathcal{C}_{1-\alpha}^*$,  we compute $\mathcal{C}_{1-\alpha}$ under the assumption that the true values for $\tau_N$, $\Sigma_N$ and $\Sigma_T$ are known. Therefore, randomness in $\mathcal{C}_{1-\alpha}^*$ stems only from the estimation of the slope parameters. For both  the feasible and the infeasible confidence interval, we report their average lengths $L(\hat{\mathcal{C}}_{1-\alpha})$ and $L(\mathcal{C}_{1-\alpha}^*)$, respectively. Across all simulations, $\alpha=0.05$.

We begin by simulating homogeneous slope parameters with spherical errors by setting $\beta_i=1$ for all individuals, $\Sigma_N=\mathrm{I}_{N\times N}$ and $\Sigma_T=\mathrm{I}_{T\times T}$ (Table \ref{sim:homogeneous}). Next, we analyze the properties of our CI under heterogeneous slopes, i.e., $\beta_i=1$ for the first half of the sample and $\beta_i=2$ for the remainder, while maintaining spherical errors (Table \ref{sim:heterogeneous}). This is followed by  results for homogeneous and heterogeneous slopes under serial correlation in the model errors (Tables \ref{sim:homogeneous-ar1phi-03} and \ref{sim:heterogeneous-ar1-phi03}), where $\varepsilon_{i,t}=\phi\varepsilon_{i,t-1} + u_{i,t}$ with $\phi=0.3$ and $u_{it}\sim\mathcal{N}(0,1)$. Consequently, $\Sigma_T$ is a Toeplitz matrix with $(\Sigma_{T})_{s,t}=\frac{1}{1-\phi^2}\phi^{|t-s|}$. Motivated by the upper bound on $\rho$ in Lemma \ref{lem_4}, the bandwidth is set to $b=T^{2/7}$ in the simulations with dynamic errors. Finally, we report results for a data generating process containing unobserved individual fixed effects (Tables \ref{sim:homogenoeus-fixedeffect} and \ref{sim:heterogeneous-fixed}), i.e., $Y_{i,t}=x_{i,t}'\beta_i+\alpha_i+\varepsilon_{i,t}$, where $\alpha_i\sim\mathcal{N}(\bar{x}_i,1)$ and $\bar{x}_i=T^{-1}\sum_{t=1}^T x_{i,t}$. In this case, our methodology can be applied by replacing $x_{i,t}$ with the demeaned version $\ddot{x}_{i,t}=x_{i,t}-\bar{x}$ and by adapting the estimator of the error covariance matrix to the demeaning operation by multiplying $M_T=\mathbf{I}_{T}- \mathbf{1}_T\mathbf{1}_T'/T$, where $\mathbf{1}_T$ denotes a vector of ones, from both sides to $\hat{\Sigma}^{(i,i)}(b)$. 

In the Appendix,  we report simulation results for heteroskedastic model errors, which are not covered by our theory (Tables \ref{sim:heteroskedastic-homogeneous} - \ref{sim:heteroskedastic-heterogeneous-parametric}). Here, we set $\mathrm{var}(\varepsilon_{i,t})= |(x_{i,t})_1|$, i.e., the variance of $\varepsilon_{i,t}$ depends on the absolute value of the first component of $x_{i,t}$. We then illustrate that our CI can be made robust to heteroskedasticity by adapting the estimator of the covariance matrix of the model errors to this feature of the data generating process. Furthermore, we apply a similar parametric estimation strategy to models with dynamic model errors and show that this approach can lead to a substantial improvement in the empirical coverage rate of your CI. While we do not provide a formal proof, we further illustrate that our CI can be made robust to both serial correlation and heteroskedasticity in the model errors by using a heteroskedasticity and autocorrelation consistent (HAC) type estimator for the covariance matrix of the model errors. Finally, we present results for the heterogeneous model with individual fixed effects, where the slope coefficients $\beta_i$ are $\mathrm{i.i.d.}$ draws from a standard normal distribution (see Table \ref{sim:heterogeneous-fixedeffect-random}).  We find that our CI performs well, which is not surprising, as our theory does not impose any distributional assumptions on the slope coefficients. 
\subsection{Results}
Across all simulations, the empirical coverage rate of the infeasible confidence interval $\mathcal{C}_{0.95}^*$ practically coincides with the desired nominal level so that  $L(\mathcal{C}_{0.95}^*)$ provides a sensible benchmark for the length of the feasible confidence interval. 
Table \ref{sim:homogeneous} illustrates that our feasible confidence interval is conservative in the homogeneous slopes model when the model errors are uncorrelated. Unsurprisingly, it is thus substantially wider than the infeasible CI for any sample size. We find that this is due to the overestimation of the conditional asymptotic variance $\tau_N^2$ under slope homogeneity, as $\Lambda=\Lambda_k=0$ in \eqref{e:deftau} whereas $\hat  \Lambda$ and $\hat  \Lambda_i$ are non-zero in \eqref{e:def-tauhat}. Arguably, when choosing between the pooled and the heterogeneous model for prediction, being conservative under slope homogeneity is somewhat desirable, as it implies that the right endpoint of the CI is rarely below zero, correctly indicating that the pooled estimator is preferable in terms of the prediction error. When slope coefficients are heterogeneous, the feasible CI closely approximates the infeasible CI, as shown in Table \ref{sim:heterogeneous}, even when $T$ is only moderately large.  A similar behavior can be observed when the model errors follow an AR(1) process (see Tables \ref{sim:homogeneous-ar1phi-03} and \ref{sim:heterogeneous-ar1-phi03}), albeit with the additional requirement that $N$ should not be too large relative to $T$, as expected in the presence of time dependence. For instance, Tables \ref{sim:homogeneous-ar1phi-03} and \ref{sim:heterogeneous-ar1-phi03} show that  $\hat{\mathcal{C}}_{0.95}$ can undercover the true difference of prediction errors when $N$ is very large relative to $T$, since estimation error in $\hat{\Sigma}^{(i,i)}(b)$ affects the accuracy of the point estimate $\hat{E}-2\hat{E}_1$. However, $T$ only needs to be moderately large relative to $N$ (e.g., $N=500$ and $T\approx 30$), to make this effect negligible. As we illustrate in the Appendix, distortions in the empirical coverage rate can be more severe with higher levels of serial correlation (see Tables \ref{sim:homogeneous-ar1phi-05} and \ref{sim:heterogeneous-ar1-phi05}) so that longer panels might be necessary for a satisfactory coverage rate. However, we also provide numerical evidence that the empirical coverage rate can be further improved when a parametric estimator is used to estimate the covariance matrix of the model errors instead of the nonparametric estimator $\hat{\Sigma}^{(i,i)}(b)$ (see Tables \ref{sim:homogeneous-parametric-ar1phi-05} and \ref{sim:heterogeneous-parametric-ar1-phi05}).
When the model errors are independent but the model contains an unobserved individual fixed effect, the data must be demeaned across time, leading to serial correlation in the demeaned model errors. Consequently, the results in Tables \ref{sim:homogenoeus-fixedeffect} and \ref{sim:heterogeneous-fixed} resemble the ones in \ref{sim:homogeneous-ar1phi-03} and \ref{sim:heterogeneous-ar1-phi03}.  Again, the feasible CI undercovers the true difference in prediction errors only when $N$ is very large relative to $T$. However, already when $T$ is as small as 20, the empirical level is close to the nominal level, and the length of the feasible CI closely approximates the length of the infeasible benchmark. As illustrated by Table \ref{sim:heterogeneous-fixedeffect-random}, our approach is also robust to individual fixed effects when the slope heterogeneity is not fixed but rather follows a random coefficients specification where $\beta_i$ is drawn from a standard normal distribution. As before, the empirical coverage rate is very accurate when $T$ is  moderately large, i.e., $T\approx 20$.
\begin{table}[H]\centering
\begin{tabular}{@{}lrrrrcrrrr@{}}\toprule
& \multicolumn{4}{c}{$N=100$} & \phantom{abc}& \multicolumn{4}{c}{$N=500$} \\
\cmidrule{2-5} \cmidrule{7-10}
$T$& $\hat{\mathcal{C}}_{0.95}$ & $L(\hat{\mathcal{C}}_{0.95})$ & $\mathcal{C}_{0.95}^*$ &$L(\mathcal{C}_{0.95}^*)$ && $\hat{\mathcal{C}}_{0.95}$ & $L(\hat{\mathcal{C}}_{0.95})$ & $\mathcal{C}_{0.95}^*$ &$L(\mathcal{C}_{0.95}^*)$\\
 \midrule
10& 0.9922	&1.6420	&0.9498	&0.9006 && 0.9890	&1.2841	&0.9530	&0.7388\\
15& 0.9972	&0.8861	&0.9572	&0.4983 && 0.9968	&0.3134	&0.9510	&0.1755\\
20& 0.9980	&0.4761	&0.9566	&0.2701 && 0.9982	&0.1905	&0.9502	&0.1075\\
25& 0.9988	&0.3169	&0.9580	&0.1802 && 0.9974	&0.1454	&0.9518	&0.0826\\
30& 0.9982	&0.2660	&0.9514	&0.1517 && 0.9992	&0.1084	&0.9488	&0.0617\\
40& 0.9990	&0.1820	&0.9518	&0.1038 && 0.9986	&0.0741	&0.9526	&0.0424\\
60& 0.9990	&0.1046	&0.9538	&0.0600 && 0.9986	&0.0482	&0.9480	&0.0276\\
80& 0.9990	&0.0794	&0.9474	&0.0456 && 0.9988	&0.0339	&0.9512	&0.0195\\
\bottomrule
\end{tabular}
\caption{Coverage rates and length of feasible and infeasible confidence interval with homogeneous slope coefficients.}
\label{sim:homogeneous}
\end{table}

\begin{table}[H]\centering
\begin{tabular}{@{}lrrrrcrrrr@{}}\toprule
& \multicolumn{4}{c}{$N=100$} & \phantom{abc}& \multicolumn{4}{c}{$N=500$} \\
\cmidrule{2-5} \cmidrule{7-10}
$T$& $\hat{\mathcal{C}}_{0.95}$ & $L(\hat{\mathcal{C}}_{0.95})$ & $\mathcal{C}_{0.95}^*$ &$L(\mathcal{C}_{0.95}^*)$ && $\hat{\mathcal{C}}_{0.95}$ & $L(\hat{\mathcal{C}}_{0.95})$ & $\mathcal{C}_{0.95}^*$ &$L(\mathcal{C}_{0.95}^*)$\\
 \midrule
10& 0.9652	&3.0652	&0.9536	&2.7388 && 0.9732	&1.7765	&0.9524	&1.4231\\
15& 0.9608	&2.1541	&0.9542	&2.0317 && 0.9604	&0.8114	&0.9512	&0.7689\\
20& 0.9660	&1.5388	&0.9598	&1.4877 && 0.9552	&0.6446	&0.9466	&0.6250\\
25& 0.9564	&1.3201	&0.9554	&1.2948 && 0.9556	&0.5575	&0.9494	&0.5444\\
30& 0.9598	&1.1893	&0.9568	&1.1693 && 0.9528	&0.4808	&0.9494	&0.4725\\
40& 0.9596	&0.9772	&0.9568	&0.9668 && 0.9556	&0.3978	&0.9536	&0.3931\\
60& 0.9602	&0.7370	&0.9590	&0.7324 && 0.9514	&0.3186	&0.9490	&0.3161\\
80& 0.9538	&0.6435	&0.9528	&0.6403 && 0.9526	&0.2699	&0.9486	&0.2684\\
\bottomrule
\end{tabular}
\caption{Coverage rates and length of feasible and infeasible confidence interval with heterogeneous slope coefficients.}
\label{sim:heterogeneous}
\end{table}

\begin{table}[H]\centering
\begin{tabular}{@{}lrrrrcrrrr@{}}\toprule
& \multicolumn{4}{c}{$N=100$} & \phantom{abc}& \multicolumn{4}{c}{$N=500$} \\
\cmidrule{2-5} \cmidrule{7-10}
$T$& $\hat{\mathcal{C}}_{0.95}$ & $L(\hat{\mathcal{C}}_{0.95})$ & $\mathcal{C}_{0.95}^*$ &$L(\mathcal{C}_{0.95}^*)$ && $\hat{\mathcal{C}}_{0.95}$ & $L(\hat{\mathcal{C}}_{0.95})$ & $\mathcal{C}_{0.95}^*$ &$L(\mathcal{C}_{0.95}^*)$\\
 \midrule
10& 0.9760	&2.8087	&0.9506	&1.6053 && 0.9298	&1.5023	&0.9510	&0.8435\\
15& 0.9792	&1.1745	&0.9562	&0.6975 && 0.8740	&0.5415	&0.9512	&0.3168\\
20& 0.9910	&0.7589	&0.9600	&0.4471 && 0.9102	&0.3393	&0.9554	&0.1998\\
25& 0.9942	&0.5597	&0.9622	&0.3298 && 0.9424	&0.2485	&0.9514	&0.1460\\
30& 0.9958	&0.4463	&0.9576	&0.2627 && 0.9594	&0.1961	&0.9522	&0.1152\\
40& 0.9980	&0.3149	&0.9574	&0.1852 && 0.9782	&0.1390	&0.9502	&0.0816\\
60& 0.9994	&0.1988	&0.9574	&0.1167 && 0.9918	&0.0880	&0.9536	&0.0515\\
80& 0.9992	&0.1444	&0.9538	&0.0843 && 0.9962	&0.0647	&0.9526	&0.0378\\
\bottomrule
\end{tabular}
\caption{ Coverage rates and length of feasible and infeasible confidence interval with homogeneous slope coefficients and AR(1) model errors with $\phi=0.3$.}
\label{sim:homogeneous-ar1phi-03}
\end{table}

\begin{table}[H]\centering
\begin{tabular}{@{}lrrrrcrrrr@{}}\toprule
& \multicolumn{4}{c}{$N=100$} & \phantom{abc}& \multicolumn{4}{c}{$N=500$} \\
\cmidrule{2-5} \cmidrule{7-10}
$T$& $\hat{\mathcal{C}}_{0.95}$ & $L(\hat{\mathcal{C}}_{0.95})$ & $\mathcal{C}_{0.95}^*$ &$L(\mathcal{C}_{0.95}^*)$ && $\hat{\mathcal{C}}_{0.95}$ & $L(\hat{\mathcal{C}}_{0.95})$ & $\mathcal{C}_{0.95}^*$ &$L(\mathcal{C}_{0.95}^*)$\\
 \midrule
10& 0.9592	&4.1472	&0.9494	&3.5791 && 0.9366	&2.0391	&0.9480	&1.6667\\
15& 0.9494	&2.4120	&0.9516	&2.3537 && 0.9184	&1.0872	&0.9482	&1.0488\\
20& 0.9512	&1.9234	&0.9536	&1.9169 && 0.9332	&0.8457	&0.9500	&0.8415\\
25& 0.9494	&1.6337	&0.9504	&1.6393 && 0.9334	&0.7218	&0.9494	&0.7243\\
30& 0.9534	&1.4501	&0.9560	&1.4603 && 0.9392	&0.6382	&0.9492	&0.6427\\
40& 0.9498	&1.2086	&0.9548	&1.2217 && 0.9422	&0.5380	&0.9508	&0.5434\\
60& 0.9486	&0.9627	&0.9522	&0.9749 && 0.9454	&0.4291	&0.9482	&0.4343\\
80& 0.9554	&0.8286	&0.9586	&0.8364 && 0.9466	&0.3674	&0.9516	&0.3709\\
\bottomrule
\end{tabular}
\caption{Coverage rates and length of feasible and infeasible confidence interval with heterogeneous slope coefficients and AR(1) model errors with $\phi=0.3$.}
\label{sim:heterogeneous-ar1-phi03}
\end{table}

\begin{table}[H]\centering
\begin{tabular}{@{}lrrrrcrrrr@{}}\toprule
& \multicolumn{4}{c}{$N=100$} & \phantom{abc}& \multicolumn{4}{c}{$N=500$} \\
\cmidrule{2-5} \cmidrule{7-10}
$T$& $\hat{\mathcal{C}}_{0.95}$ & $L(\hat{\mathcal{C}}_{0.95})$ & $\mathcal{C}_{0.95}^*$ &$L(\mathcal{C}_{0.95}^*)$ && $\hat{\mathcal{C}}_{0.95}$ & $L(\hat{\mathcal{C}}_{0.95})$ & $\mathcal{C}_{0.95}^*$ &$L(\mathcal{C}_{0.95}^*)$\\
 \midrule
10& 0.8386	&5.0807	&0.9562	&3.3365 && 0.3550	&2.4451	&0.9464	&1.5662\\
15& 0.9790	&1.4493	&0.9620	&0.8762 && 0.7832	&0.6764	&0.9480	&0.4065\\
20& 0.9910	&0.9164	&0.9522	&0.5432 && 0.9422	&0.4005	&0.9576	&0.2373\\
25& 0.9944	&0.6681	&0.9534	&0.3937 && 0.9742	&0.2821	&0.9508	&0.1658\\
30& 0.9988	&0.4909	&0.9570	&0.2872 && 0.9872	&0.2171	&0.9528	&0.1271\\
40& 0.9990	&0.3455	&0.9568	&0.2015 && 0.9940	&0.1513	&0.9496	&0.0881\\
60& 0.9990	&0.2076	&0.9596	&0.1209 && 0.9976	&0.0941	&0.9464	&0.0546\\
80& 0.9996	&0.1485	&0.9482	&0.0862 && 0.9992	&0.0678	&0.9564	&0.0393\\
\bottomrule
\end{tabular}
\caption{Coverage rates and length of feasible and infeasible confidence interval with homogeneous slopes and individual fixed effects. }
\label{sim:homogenoeus-fixedeffect}
\end{table}

\begin{table}[H]\centering
\begin{tabular}{@{}lrrrrcrrrr@{}}\toprule
& \multicolumn{4}{c}{$N=100$} & \phantom{abc}& \multicolumn{4}{c}{$N=500$} \\
\cmidrule{2-5} \cmidrule{7-10}
$T$& $\hat{\mathcal{C}}_{0.95}$ & $L(\hat{\mathcal{C}}_{0.95})$ & $\mathcal{C}_{0.95}^*$ &$L(\mathcal{C}_{0.95}^*)$ && $\hat{\mathcal{C}}_{0.95}$ & $L(\hat{\mathcal{C}}_{0.95})$ & $\mathcal{C}_{0.95}^*$ &$L(\mathcal{C}_{0.95}^*)$\\
 \midrule
10& 0.8894	&8.5260	&0.9544	&7.7738 && 0.6284	&4.8723	&0.9494	&3.9320\\
15& 0.9522	&4.3905	&0.9540	&4.1887 && 0.9002	&1.9939	&0.9504	&1.8926\\
20& 0.9510	&3.1799	&0.9496	&3.0716 && 0.9446	&1.4991	&0.9512	&1.4498\\
25& 0.9560	&2.8129	&0.9520	&2.7550 && 0.9514	&1.2308	&0.9538	&1.2063\\
30& 0.9556	&2.3995	&0.9532	&2.3718 && 0.9500	&1.0865	&0.9468	&1.0696\\
40& 0.9562	&1.9759	&0.9558	&1.9589 && 0.9482	&0.8913	&0.9474	&0.8818\\
60& 0.9518	&1.5410	&0.9524	&1.5321 && 0.9504	&0.7028	&0.9488	&0.6984\\
80& 0.9492	&1.3044	&0.9490	&1.3010 && 0.9512	&0.5944	&0.9518	&0.5918\\
\bottomrule
\end{tabular}
\caption{Coverage rates and length of feasible and infeasible confidence interval with slope heterogeneity and individual fixed effects.}
\label{sim:heterogeneous-fixed}
\end{table}

\section{Conclusion}
Researchers frequently face a bias-variance trade-off when choosing between pooled and individual-specific estimators for the slope coefficients in panel data analysis. A sensible strategy is to use the estimator that yields a smaller mean squared prediction error. Here, we have derived a closed form expression of the difference in mean squared prediction errors between the pooled and the individual-specific OLS estimators for panel data models with potentially heterogeneous slopes. We have then constructed a novel confidence interval for the said difference in mean squared prediction errors. Our asymptotic analysis shows that our confidence interval has asymptotically the correct coverage as $N,T\to\infty$, while allowing for the cross-sectional dimension to grow at a faster rate than the time-dimension. By means of an extensive simulation study, we have demonstrated that the empirical coverage rate is close to the nominal coverage rate in sufficiently long panels, even when the cross-sectional dimension is much larger than the time series dimension. Finally, we have illustrated that our confidence interval can be flexibly adapted to features of the data generating process to further enhance its small-sample performance. 

{\bf Acknowledgments. } 
Holger Dette has been partially supported
by the Deutsche Forschungsgemeinschaft (DFG), project number 45723897,  and by  TRR
391 {\it Spatio-temporal Statistics for the Transition of Energy and Transport}, project number
520388526 (DFG).
 Tim Kutta's work has been partially funded by AUFF grants 47331 and 47222.

\clearpage                
\bibliographystyle{chicago} 
\bibliography{ref_separability-supnorm}

\newpage 

\appendix

\begin{center}
\LARGE \textbf{Appendix}
\end{center}

\section{Additional simulation results}\label{appendix:additional-sim}
In this section, we present additional simulation results complementing our main simulations. As mentioned in Section \ref{section-simulations}, the empirical coverage rate of our CI depends on the accuracy of the estimator of the covariance matrix of the model errors $\hat{\Sigma}^{(i,k)}(b)$. While we do not provide explicit proofs for each case, we illustrate by means of a simulation study that using estimators of the error covariance matrix that are specifically tailored to features of the data generating process leads to substantially improved coverage of our CI. Through this strategy, our CI can be made more reliable even in "short" panel data sets, where the cross-sectional dimension is much larger than the time series dimension. 

We first consider the case where the  functional form of heteroskedasticity is explicitly modeled. Assume that the model error can be written as $\varepsilon_{i,t} = \omega(x_{i,t}) \cdot u_{i,t}$, where $\omega(\cdot)$ is a known positive scaling function and $u_{i,t}$ is a zero-mean, homoskedastic process with strictly temporal correlation. The error covariance matrix can then be estimated as follows. First, the standardized residuals $\hat{\varepsilon}_{i,t} = \hat{u}_{i,t}/\omega(x_{i,t})$ are computed. Since the standardized residuals are homoskedastic, we can then compute $\hat{\Sigma}^{(i,i)}(b)$ as in \eqref{def_hat_Sigma_i}. The final estimator of the covariance of the model errors is 
\begin{equation}\label{def:Sigmahat-heteroskedasticity}
\tilde{\hat{\Sigma}}^{(i,i)}(b)=\Omega_i\hat\Sigma^{(i,i)}(b)\Omega_i,
\end{equation}
where $\Omega_i=\mathrm{diag}(\omega(x_{i,1}), \dots, \omega(x_{i,T}))$. Table \ref{sim:heteroskedastic-homogeneous} shows that if the presence of heteroskedasticity is not accounted for, the CI can undercover when $N$ is large and slopes are homogeneous. This is not surprising as heteroskedasticity violates Assumption \ref{ass_2}. However, under heterogeneous slopes, a sufficiently large time series appears to yield approximately accurate coverage irrespective of the cross-sectional dimension (Table \ref{sim:heteroskedastic-heterogeneous}). An intuitive explanation for this finding is that the additional variation in the errors due to heteroskedasticity (which is not captured by $\hat{\Sigma}^{(i,i)}(b)$ in \eqref{def_hat_Sigma_i}) matters less when most of the variation in the data stems from heterogeneous slope coefficients. In contrast, the corresponding results where $\hat{\Sigma}^{(i,i)}(b)$ is replaced by $\tilde{\hat{\Sigma}}^{(i,i)}(b)$  (Tables \ref{sim:heteroskedastic-homogeneous-parametric} and \ref{sim:heteroskedastic-heterogeneous-parametric}) show that the heteroskedasticity-robust version of the CI performs very well, even when $T$ is small relative to $N$. Intuitively, $\tilde{\hat{\Sigma}}^{(i,i)}(b)$ correctly incorporates additional variation due to heteroskedasticity, leading to reliable inference in short panels irrespective of the nature of the slope coefficients.

Next, we illustrate that a similar parametric estimation strategy improves the empirical coverage rate of our CI for small $T$ when the errors follow an AR(1) process with high levels of serial correlation. Since the model errors are generated as $\varepsilon_{i,t}=\phi\varepsilon_{i,t-1}+u_{i,t}$, we first obtain the least squares estimator $\hat{\phi}$ from a regression of the residuals on the lagged residuals. As $\hat{\phi}$ suffers from small-sample bias in short time series, we compute the bias-corrected estimator $\hat{\phi}_{bc}=\hat{\phi}+(1+3\hat{\phi})/T$ (see \citealp{marriott1954bias} and \citealp{kendall1954note}). Next, we estimate $\sigma_{u_i}^2=\mathbb{V}ar(u_{i,t})$ by the sample variance of the residuals, i.e., $\hat{\sigma}_{u, i}^2 = \frac{1}{T-K} \sum_{t=1}^T \hat{\varepsilon}_{i,t}^2$. Finally, the estimated covariance matrix of the model errors is obtained as the Toeplitz matrix with entries
\begin{equation}\label{def:Sigmahat-ar1}
    \dot{\hat{\Sigma}}^{(i,i)}_{s,t}=\hat{\sigma}_{u, i}^2\, \hat{\phi}_i^{|s-t|}
\end{equation}
for $s,t\in\{1,...,T\}$. Our results show that replacing the nonparametric kernel estimator $\hat{\Sigma}^{(i,i)}(b)$ with the parametric $\dot{\hat{\Sigma}}^{(i,i)}$ that (correctly) imposes an AR(1) structure on the errors yields a notable improvement in the performance of the feasible confidence interval. As shown in Tables \ref{sim:homogeneous-ar1phi-05} and \ref{sim:heterogeneous-ar1-phi05}, where the nonparametric estimator $\hat{\Sigma}^{(i,i)}(b)$ is used, the empirical coverage rate can be substantially lower than the nominal rate if $N$ is too large relative to $T$. Therefore, when serial correlation in the model errors is strong, a longer time series is required for accurate inference when the nonparametric kernel estimator is used. In contrast. Tables \ref{sim:homogeneous-parametric-ar1phi-05} and \ref{sim:heterogeneous-parametric-ar1-phi05} show that the empirical coverage is acceptable already when $T\approx 25$ when the covariance matrix of the model errors is estimated using the parametric estimator $\dot{\hat{\Sigma}}^{(i,i)}$.

Finally, we illustrate that it is not necessary to impose a parametric specification on the model errors in the presence of both heteroskedasticity and serial correlation. However, we also show that nonparametric estimation comes at a cost, as much longer time dimensions are necessary for satisfactory coverage. To do so, we generate $\varepsilon_{i,t}= |(x_{i,t})_1| u_{i,t}$ and $u_{i,t}=0.3u_{i,t-1}+\eta_{i,t}$ with $\eta_{i,t}$ drawn independently from a standard normal distribution. Instead of using a parametric specification of the heteroskedasticity and  serial correlation as before, we employ a heteroskedasticity and autocorrelation consistent (HAC) type of estimator (Tables \ref{sim:hac-homogeneous} and \ref{sim:hac-heterogeneous}). Recall the OLS residual  $\hat{\varepsilon}_{i,t} = y_{i,t} - x_{i,t}'\hat{\beta}_i$ for individual $i$ at time $t$. The estimator $\hat{\Sigma}_{HAC}^{(i,i)}(b)$ has elements
\begin{equation} \label{eq:sigma_hac_element}
    (\hat{\Sigma}_{HAC}^{(i,i)}(b))_{s,t} = w_{s,t}(b)\, \hat{\varepsilon}_{i,s} \hat{\varepsilon}_{i,t}, \quad \quad 1 \le s,t \le T,
\end{equation}
where $w_{s,t}(b)=(1 - \frac{|s-t|}{b+1})\mathbb{I}\{|s-t| < b\}$ are the Bartlett Kernel weights as in \cite{newey-west}. In our simulations, we choose the bandwidth parameter $b=T^{2/7}$. The simulation results in Tables  \ref{sim:hac-homogeneous} and \ref{sim:hac-heterogeneous} show that the CI based on the HAC estimator works well when $T$ is large. Our results therefore suggest that our CI can be made robust to both serial correlation and heteroskedasticity of unknown form when a sufficiently long time series is available.

\begin{table}[H]\centering
\begin{tabular}{@{}lrrrrcrrrr@{}}\toprule
& \multicolumn{4}{c}{$N=100$} & \phantom{abc}& \multicolumn{4}{c}{$N=500$} \\
\cmidrule{2-5} \cmidrule{7-10}
$T$& $\hat{\mathcal{C}}_{0.95}$ & $L(\hat{\mathcal{C}}_{0.95})$ & $\mathcal{C}_{0.95}^*$ &$L(\mathcal{C}_{0.95}^*)$ && $\hat{\mathcal{C}}_{0.95}$ & $L(\hat{\mathcal{C}}_{0.95})$ & $\mathcal{C}_{0.95}^*$ &$L(\mathcal{C}_{0.95}^*)$\\
 \midrule
10& 0.9784	&1.8746	&0.9528	&1.0744 && 0.9894	&1.2405	&0.9556	&0.6154 \\
15& 0.9982	&1.0523	&0.9618	&0.5985 && 0.9384	&0.3662	&0.9530	&0.2125 \\
20& 0.9980	&0.5339	&0.9526	&0.3119 && 0.9004	&0.2304	&0.9480	&0.1375\\
25& 0.9970	&0.3806	&0.9600	&0.2265 && 0.8808	&0.1706	&0.9528	&0.1037\\
30& 0.9984	&0.3196	&0.9598	&0.1888 && 0.8450	&0.1288	&0.9500	&0.0776\\
40& 0.9982	&0.2112	&0.9560	&0.1266 && 0.8028	&0.0886	&0.9524	&0.0543\\
60& 0.9956	&0.1256	&0.9588	&0.0772 && 0.7430	&0.0570	&0.9458	&0.0352\\
80& 0.9982	&0.0961	&0.9556	&0.0586 && 0.7596	&0.0406	&0.9508	&0.0253\\
100& 0.9986	&0.0720	&0.9506	&0.0434 && 0.7876	&0.0318	&0.9482	&0.0196\\
\bottomrule
\end{tabular}
\caption{Coverage rates and length of feasible and infeasible confidence interval with homogeneous slope coefficients and heteroskedastic errors.}
\label{sim:heteroskedastic-homogeneous}
\end{table}

\begin{table}[H]\centering
\begin{tabular}{@{}lrrrrcrrrr@{}}\toprule
& \multicolumn{4}{c}{$N=100$} & \phantom{abc}& \multicolumn{4}{c}{$N=500$} \\
\cmidrule{2-5} \cmidrule{7-10}
$T$& $\hat{\mathcal{C}}_{0.95}$ & $L(\hat{\mathcal{C}}_{0.95})$ & $\mathcal{C}_{0.95}^*$ &$L(\mathcal{C}_{0.95}^*)$ && $\hat{\mathcal{C}}_{0.95}$ & $L(\hat{\mathcal{C}}_{0.95})$ & $\mathcal{C}_{0.95}^*$ &$L(\mathcal{C}_{0.95}^*)$\\
 \midrule
10& 0.9604	&3.2795	&0.9568	&2.9978 && 0.9762	&1.7721	&0.9516	&1.3960 \\
15& 0.9590	&2.3875	&0.9494	&2.2371 && 0.9340	&0.8704	&0.9486	&0.8548 \\
20& 0.9620	&1.6265	&0.9610	&1.6196 && 0.9324	&0.6993	&0.9462	&0.7105\\
25& 0.9510	&1.4247	&0.9530	&1.4429 && 0.9348	&0.5981	&0.9494	&0.6109\\
30& 0.9598	&1.2904	&0.9594	&1.2895 && 0.9348	&0.5151	&0.9532	&0.5331\\
40& 0.9530	&1.0263	&0.9594	&1.0592 && 0.9354	&0.4279	&0.9532	&0.4456\\
60& 0.9560	&0.7980	&0.9608	&0.8282 && 0.9318	&0.3429	&0.9512	&0.3604\\
80& 0.9484	&0.7024	&0.9544	&0.7242 && 0.9306	&0.2902	&0.9480	&0.3059\\
100& 0.9510	&0.6032	&0.9582	&0.6256 && 0.9370	&0.2576	&0.9526	&0.2715\\
\bottomrule
\end{tabular}
\caption{Coverage rates and length of feasible and infeasible confidence interval with heterogeneous slope coefficients and heteroskedastic errors.}
\label{sim:heteroskedastic-heterogeneous}
\end{table}

\begin{table}[H]\centering
\begin{tabular}{@{}lrrrrcrrrr@{}}\toprule
& \multicolumn{4}{c}{$N=100$} & \phantom{abc}& \multicolumn{4}{c}{$N=500$} \\
\cmidrule{2-5} \cmidrule{7-10}
$T$& $\hat{\mathcal{C}}_{0.95}$ & $L(\hat{\mathcal{C}}_{0.95})$ & $\mathcal{C}_{0.95}^*$ &$L(\mathcal{C}_{0.95}^*)$ && $\hat{\mathcal{C}}_{0.95}$ & $L(\hat{\mathcal{C}}_{0.95})$ & $\mathcal{C}_{0.95}^*$ &$L(\mathcal{C}_{0.95}^*)$\\
 \midrule
10& 0.9902	&2.7163	&0.9534	&1.4944 && 0.9876	&1.4680	&0.9508	&0.8058 \\
15& 0.9948	&1.0852	&0.9548	&0.6194 && 0.9886	&0.4993	&0.9514	&0.2831 \\
20& 0.9982	&0.6805	&0.9560	&0.3890 && 0.9962	&0.3029	&0.9572	&0.1729\\
25& 0.9984	&0.4902	&0.9574	&0.2811 && 0.9964	&0.2187	&0.9488	&0.1252\\
30& 0.9992	&0.3873	&0.9590	&0.2228 && 0.9976	&0.1702	&0.9498	&0.0978\\
40& 0.9996	&0.2680	&0.9568	&0.1546 && 0.9990	&0.1185	&0.9490	&0.0682\\
60& 0.9988	&0.1664	&0.9564	&0.0965 && 0.9992	&0.0737	&0.9524	&0.0425\\
\bottomrule
\end{tabular}
\caption{Coverage rates and length of feasible and infeasible confidence interval with homogeneous slope coefficients and heteroskedastic errors. $\hat{\Sigma}^{(i,i)}(b)$ is replaced with $\tilde{\hat{\Sigma}}^{(i,i)}(b)$ in  \eqref{def:Sigmahat-heteroskedasticity}.}
\label{sim:heteroskedastic-homogeneous-parametric}
\end{table}

\begin{table}[H]\centering
\begin{tabular}{@{}lrrrrcrrrr@{}}\toprule
& \multicolumn{4}{c}{$N=100$} & \phantom{abc}& \multicolumn{4}{c}{$N=500$} \\
\cmidrule{2-5} \cmidrule{7-10}
$T$& $\hat{\mathcal{C}}_{0.95}$ & $L(\hat{\mathcal{C}}_{0.95})$ & $\mathcal{C}_{0.95}^*$ &$L(\mathcal{C}_{0.95}^*)$ && $\hat{\mathcal{C}}_{0.95}$ & $L(\hat{\mathcal{C}}_{0.95})$ & $\mathcal{C}_{0.95}^*$ &$L(\mathcal{C}_{0.95}^*)$\\
 \midrule
10& 0.9680	&4.0582	&0.9466	&3.3854 && 0.9716	&2.0019	&0.9464	&1.5865 \\
15& 0.9576	&2.2923	&0.9500	&2.1557 && 0.9562	&1.0295	&0.9462	&0.9604 \\
20& 0.9568	&1.7938	&0.9556	&1.7325 && 0.9570	&0.7872	&0.9500	&0.7592\\
25& 0.9546	&1.5033	&0.9490	&1.4664 && 0.9554	&0.6651	&0.9524	&0.6490\\
30& 0.9578	&1.3266	&0.9550	&1.3032 && 0.9498	&0.5833	&0.9464	&0.5729\\
40& 0.9536	&1.0928	&0.9524	&1.0811 && 0.9516	&0.4863	&0.9480	&0.4807\\
60& 0.9508	&0.8603	&0.9506	&0.8580 && 0.9476	&0.3829	&0.9464	&0.3817\\
\bottomrule
\end{tabular}
\caption{Coverage rates and length of feasible and infeasible confidence interval with heterogeneous slope coefficients and heteroskedastic errors. $\hat{\Sigma}^{(i,i)}(b)$ is replaced with $\tilde{\hat{\Sigma}}^{(i,i)}(b)$ in  \eqref{def:Sigmahat-heteroskedasticity}. }
\label{sim:heteroskedastic-heterogeneous-parametric}
\end{table}

\begin{table}[H]\centering
\begin{tabular}{@{}lrrrrcrrrr@{}}\toprule
& \multicolumn{4}{c}{$N=100$} & \phantom{abc}& \multicolumn{4}{c}{$N=500$} \\
\cmidrule{2-5} \cmidrule{7-10}
$T$& $\hat{\mathcal{C}}_{0.95}$ & $L(\hat{\mathcal{C}}_{0.95})$ & $\mathcal{C}_{0.95}^*$ &$L(\mathcal{C}_{0.95}^*)$ && $\hat{\mathcal{C}}_{0.95}$ & $L(\hat{\mathcal{C}}_{0.95})$ & $\mathcal{C}_{0.95}^*$ &$L(\mathcal{C}_{0.95}^*)$\\
 \midrule
10& 0.9280	&1.8027	&0.9544	&1.0244 && 0.9482	&1.2874	&0.9580	&0.6929\\
15& 0.8974	&1.1248	&0.9590	&0.6848 && 0.8186	&0.3381	&0.9506	&0.1988\\
20& 0.9070	&0.6659	&0.9564	&0.4103 && 0.9100	&0.2095	&0.9472	&0.1231\\
25& 0.9450	&0.4441	&0.9590	&0.2703 && 0.9450	&0.1638	&0.9454	&0.0961\\
30& 0.9640	&0.3810	&0.9506	&0.2330 && 0.9502	&0.1255	&0.9488	&0.0739\\
40& 0.9688	&0.2786	&0.9500	&0.1708 && 0.9814	&0.0859	&0.9504	&0.0503\\
60& 0.9942	&0.1656	&0.9498	&0.0990 && 0.9900	&0.0570	&0.9480	&0.0333\\
80& 0.9944	&0.1300	&0.9518	&0.0778 && 0.9952	&0.0403	&0.9484	&0.0235\\
\bottomrule
\end{tabular}
\caption{Coverage rates and length of feasible and infeasible confidence interval with homogeneous slope coefficients and AR(1) model errors with $\phi=0.5$.}
\label{sim:homogeneous-ar1phi-05}
\end{table}

\begin{table}[H]\centering
\begin{tabular}{@{}lrrrrcrrrr@{}}\toprule
& \multicolumn{4}{c}{$N=100$} & \phantom{abc}& \multicolumn{4}{c}{$N=500$} \\
\cmidrule{2-5} \cmidrule{7-10}
$T$& $\hat{\mathcal{C}}_{0.95}$ & $L(\hat{\mathcal{C}}_{0.95})$ & $\mathcal{C}_{0.95}^*$ &$L(\mathcal{C}_{0.95}^*)$ && $\hat{\mathcal{C}}_{0.95}$ & $L(\hat{\mathcal{C}}_{0.95})$ & $\mathcal{C}_{0.95}^*$ &$L(\mathcal{C}_{0.95}^*)$\\
 \midrule
10& 0.9334	&3.2449	&0.9580	&3.1588  && 0.8632	&1.8881	&0.9514	&1.5780\\
15& 0.9236	&2.4419	&0.9550	&2.5324 && 0.7856	&0.9347	&0.9504	&0.9803\\
20& 0.9250	&1.8296	&0.9590	&1.9552	&& 0.8508	&0.7578	&0.9500	&0.8072\\
25& 0.9338	&1.5833	&0.9530	&1.6796	&& 0.8886	&0.6715	&0.9504	&0.7119\\
30& 0.9368	&1.4651	&0.9592	&1.5604 && 0.8940	&0.5959	&0.9482	&0.6345\\
40& 0.9366	&1.2333	&0.9586	&1.3205 && 0.9216	&0.5001	&0.9546	&0.5320\\
60& 0.9484	&0.9570	&0.9618	&1.0009 && 0.9336	&0.4161	&0.9472	&0.4355\\
80& 0.9500	&0.8526	&0.9582	&0.8863 && 0.9362	&0.3555	&0.9512	&0.3699\\
\bottomrule
\end{tabular}
\caption{Coverage rates and length of feasible and infeasible confidence interval with heterogeneous slope coefficients and AR(1) model errors with $\phi=0.5$.}
\label{sim:heterogeneous-ar1-phi05}
\end{table}

\begin{table}[H]\centering
\begin{tabular}{@{}lrrrrcrrrr@{}}\toprule
& \multicolumn{4}{c}{$N=100$} & \phantom{abc}& \multicolumn{4}{c}{$N=500$} \\
\cmidrule{2-5} \cmidrule{7-10}
$T$& $\hat{\mathcal{C}}_{0.95}$ & $L(\hat{\mathcal{C}}_{0.95})$ & $\mathcal{C}_{0.95}^*$ &$L(\mathcal{C}_{0.95}^*)$ && $\hat{\mathcal{C}}_{0.95}$ & $L(\hat{\mathcal{C}}_{0.95})$ & $\mathcal{C}_{0.95}^*$ &$L(\mathcal{C}_{0.95}^*)$\\
 \midrule
10& 0.9040	&3.1536	&0.9504	&1.9124 && 0.5436	&1.6283	&0.9516	&0.9765\\
15& 0.9586	&1.4680	&0.9562	&0.8803 && 0.7598	&0.6796	&0.9488	&0.3985\\
20& 0.9854	&1.0007	&0.9582	&0.5797 && 0.8662	&0.9790	&0.9500	&0.9922\\
25& 0.9930	&0.7564	&0.9612	&0.4362 && 0.9516	&0.3354	&0.9494	&0.1926\\
30& 0.9950	&0.6072	&0.9568	&0.3493 && 0.9758	&0.2679	&0.9542	&0.1535\\
40& 0.9978	&0.4360	&0.9570	&0.2506 && 0.9908	&0.1927	&0.9452	&0.1106\\
60& 0.9982	&0.2774	&0.9522	&0.1599 && 0.9972	&0.1229	&0.9532	&0.0706\\
\bottomrule
\end{tabular}
\caption{Coverage rates and length of feasible and infeasible confidence interval with homogeneous slope coefficients and AR(1) model errors with $\phi=0.5$. $\hat{\Sigma}^{(i,i)}(b)$ is replaced by $\dot{\hat{\Sigma}}^{(i,i)}$ in \eqref{def:Sigmahat-ar1}.}
\label{sim:homogeneous-parametric-ar1phi-05}
\end{table}

\begin{table}[H]\centering
\begin{tabular}{@{}lrrrrcrrrr@{}}\toprule
& \multicolumn{4}{c}{$N=100$} & \phantom{abc}& \multicolumn{4}{c}{$N=500$} \\
\cmidrule{2-5} \cmidrule{7-10}
$T$& $\hat{\mathcal{C}}_{0.95}$ & $L(\hat{\mathcal{C}}_{0.95})$ & $\mathcal{C}_{0.95}^*$ &$L(\mathcal{C}_{0.95}^*)$ && $\hat{\mathcal{C}}_{0.95}$ & $L(\hat{\mathcal{C}}_{0.95})$ & $\mathcal{C}_{0.95}^*$ &$L(\mathcal{C}_{0.95}^*)$\\
 \midrule
10& 0.9224	&4.5282	&0.9514	&4.0328 && 0.7526	&2.1923	&0.9468	&1.8693\\
15& 0.9428	&2.8311	&0.9536	&2.7314 && 0.8898	&1.2851	&0.9480	&1.2185\\
20& 0.9534	&2.3211	&0.9538	&2.2558	&& 0.9334	&1.0264	&0.9500	&0.9922\\
25& 0.9528	&1.9920	&0.9496	&1.9465	&& 0.9390	&0.8791	&0.9498	&0.8588\\
30& 0.9554	&1.7696	&0.9550	&1.7372 && 0.9438	&0.7811	&0.9488	&0.7655\\
40& 0.9524	&1.4834	&0.9528	&1.4649 && 0.9466	&0.6610	&0.9504	&0.6519\\
60& 0.9512	&1.1823	&0.9512	&1.1750 && 0.9488	&0.5282	&0.9488	&0.5241\\
\bottomrule
\end{tabular}
\caption{Coverage rates and length of feasible and infeasible confidence interval with heterogeneous slope coefficients and AR(1) model errors with $\phi=0.5$. $\hat{\Sigma}^{(i,i)}(b)$ is replaced by $\dot{\hat{\Sigma}}^{(i,i)}$ in \eqref{def:Sigmahat-ar1}.}
\label{sim:heterogeneous-parametric-ar1-phi05}
\end{table}

\begin{table}[H]\centering
\begin{tabular}{@{}lrrrrcrrrr@{}}\toprule
& \multicolumn{4}{c}{$N=100$} & \phantom{abc}& \multicolumn{4}{c}{$N=500$} \\
\cmidrule{2-5} \cmidrule{7-10}
$T$& $\hat{\mathcal{C}}_{0.95}$ & $L(\hat{\mathcal{C}}_{0.95})$ & $\mathcal{C}_{0.95}^*$ &$L(\mathcal{C}_{0.95}^*)$ && $\hat{\mathcal{C}}_{0.95}$ & $L(\hat{\mathcal{C}}_{0.95})$ & $\mathcal{C}_{0.95}^*$ &$L(\mathcal{C}_{0.95}^*)$\\
 \midrule
40&   0.6734	&0.2972	&0.9592	&0.1966 && 0.0034	&0.1326	&0.9498	&0.0869 \\
60&   0.9082	&0.1932	&0.9544	&0.1226 && 0.1416	&0.0868	&0.9474	&0.0545 \\
80&   0.9646	&0.1455	&0.9578	&0.0899 && 0.4758	&0.0652	&0.9570	&0.0402\\
100&  0.9848	&0.1138	&0.9540	&0.0692 && 0.7168	&0.0518	&0.9506	&0.0314\\
120&  0.9894	&0.0960	&0.9574	&0.0579 && 0.8472	&0.0433	&0.9552	&0.0260\\
150&  0.9958	&0.0767	&0.9542	&0.0459 && 0.9150	&0.0345	&0.9488	&0.0206\\
180&  0.9974	&0.0642	&0.9570	&0.0381 && 0.9554	&0.0287	&0.9476	&0.0171\\
\bottomrule
\end{tabular}
\caption{Coverage rates and length of feasible and infeasible confidence interval with homogeneous slope coefficients and heteroskedastic errors. $\hat{\Sigma}^{(i,i)}(b)$ is replaced by the HAC estimator in \eqref{eq:sigma_hac_element}.}
\label{sim:hac-homogeneous}
\end{table}
$ $\\[-8ex]
\begin{table}[H]\centering
\begin{tabular}{@{}lrrrrcrrrr@{}}\toprule
& \multicolumn{4}{c}{$N=100$} & \phantom{abc}& \multicolumn{4}{c}{$N=500$} \\
\cmidrule{2-5} \cmidrule{7-10}
$T$& $\hat{\mathcal{C}}_{0.95}$ & $L(\hat{\mathcal{C}}_{0.95})$ & $\mathcal{C}_{0.95}^*$ &$L(\mathcal{C}_{0.95}^*)$ && $\hat{\mathcal{C}}_{0.95}$ & $L(\hat{\mathcal{C}}_{0.95})$ & $\mathcal{C}_{0.95}^*$ &$L(\mathcal{C}_{0.95}^*)$\\
 \midrule
40 &  0.9060	&1.0924	&0.9526	&1.2183 && 0.7966	&0.4884 &0.9470	&0.5429 \\
60 &  0.9302	&0.8934	&0.9528	&0.9670 && 0.8882	&0.4001	&0.9486	&0.4317 \\
80 &  0.9380	&0.7891	&0.9554	&0.8369 && 0.9126	&0.3489	&0.9468	&0.3698\\
100&  0.9430	&0.7027	&0.9552	&0.7376 && 0.9334	&0.3131	&0.9534	&0.3281\\
120&  0.9438	&0.6444	&0.9552	&0.6707 && 0.9352	&0.2864	&0.9496	&0.2981\\
150&  0.9462	&0.5750	&0.9534	&0.5945 &&0.9410	&0.2564	&0.9528	&0.2653\\
180&  0.9492	&0.5270	&0.9558	&0.5429 &&0.9458	&0.2342	&0.9524	&0.2410\\
\bottomrule
\end{tabular}
\caption{Coverage rates and length of feasible and infeasible confidence interval with heterogeneous slope coefficients and heteroskedastic errors. $\hat{\Sigma}^{(i,i)}(b)$ is replaced by the HAC estimator in \eqref{eq:sigma_hac_element}.}
\label{sim:hac-heterogeneous}
\end{table}

\begin{table}[H]\centering
\begin{tabular}{@{}lrrrrcrrrr@{}}\toprule
& \multicolumn{4}{c}{$N=100$} & \phantom{abc}& \multicolumn{4}{c}{$N=500$} \\
\cmidrule{2-5} \cmidrule{7-10}
$T$& $\hat{\mathcal{C}}_{0.95}$ & $L(\hat{\mathcal{C}}_{0.95})$ & $\mathcal{C}_{0.95}^*$ &$L(\mathcal{C}_{0.95}^*)$ && $\hat{\mathcal{C}}_{0.95}$ & $L(\hat{\mathcal{C}}_{0.95})$ & $\mathcal{C}_{0.95}^*$ &$L(\mathcal{C}_{0.95}^*)$\\
 \midrule
10& 0.8894	&8.5260	&0.9544	&7.7738 && 0.6284	&4.8723	&0.9494	&3.9320\\
15& 0.9522	&4.3905	&0.9540	&4.1887 && 0.9002	&1.9939	&0.9504	&1.8926\\
20& 0.9510	&3.1799	&0.9496	&3.0716 && 0.9446	&1.4991	&0.9512	&1.4498\\
25& 0.9560	&2.8129	&0.9520	&2.7550 && 0.9514	&1.2308	&0.9538	&1.2063\\
30& 0.9556	&2.3995	&0.9532	&2.3718 && 0.9500	&1.0865	&0.9468	&1.0696\\
40& 0.9562	&1.9759	&0.9558	&1.9589 && 0.9482	&0.8913	&0.9474	&0.8818\\
60& 0.9518	&1.5410	&0.9524	&1.5321 && 0.9504	&0.7028	&0.9488	&0.6984\\
80& 0.9492	&1.3044	&0.9490	&1.3010 && 0.9512	&0.5944	&0.9518	&0.5918\\
\bottomrule
\end{tabular}
\caption{Coverage rates and length of feasible and infeasible confidence interval with $\beta_i\sim \mathrm{N}(1,1)$ and individual fixed effects.}
\label{sim:heterogeneous-fixedeffect-random}
\end{table}

\newpage  
\section{Proofs and technical details}

\noindent Throughout the Appendix, $C$ denotes a generic, positive constant that can change form one line to another. $C$ and all (stochastic) constants implied by Landau symbols are always independent of $N,T$, as well as individual indices $i,t$. In the following, we will frequently use a simple relation between stochastic Landau symbols and conditional stochastic Landau symbols. More precisely, for a sequence $(a_N)_N$ of positive numbers and a sequence of random variables $(z_N)_N$, that are $\mathbf{X}$-measurable, it holds that 
\begin{align}\label{e:lanXs}
z_N = \mathcal{O}_P(a_N) \,\,\,\Rightarrow \,\,\, z_N = \mathcal{O}^{|\mathbf{X}}_P(a_N).
\end{align}
Next, we recall some properties regarding matrices: First, let us denote by $\|\cdot\|_1$ the trace norm of a matrix, by $\|\cdot\|_2$ its Frobenius norm, by $\|\cdot\|_\infty$ its spectral norm, by $\|\cdot\|_{row}$ is absolute row sum norm and by $Tr[\cdot]$ the trace. Then, we recall for a matrices $A,B,C$ 
\begin{itemize}
    \item[M1)] $\|A\|_\infty \le \|A\|_2 \le \|A\|_1$ with equality  if $A$ has rank $1$.
    \item[M2)] $\|AB\|_1 \le \|A\|_1 \|B\|_\infty$.
    \item[M3)] $Tr[ABC]=Tr[CAB]$.
    \item[M4)] $Tr[A] \le \|A\|_1$, with equality if $A$ is a covariance matrix.
    \item[M5)] $\|ABC\|_p \le \|A\|_p \|B\|_p \|C\|_p $, for $p \in \{1,2, \infty\}$. 
    \item[M6)] $\|A\|_\infty \le \|A\|_{row}$ for a symmetric matrix $A$. 
\end{itemize}
Here, in conditions M3) and M5) we have assumed that the respective matrix products on both sides make sense. Property M5) can be found in \cite{Matbook}  (p.17). M6) is a special case of  Gershgorin's circle theorem and often useful when considering the normed difference of covariance matrices. \\
Finally, we come to the assumptions in our proofs. 
In Assumption \ref{ass_2}, $iii)$ we have assumed that the model errors have exponentially decaying mixing coefficients and in Assumption \ref{ass_3}, $iii)$ that $N/T^\eta\to 0$ for some $\eta \in (0,2)$. In the Appendix, we present proofs under somewhat more general conditions. Namely, we assume that the mixing coefficients satisfy for some constant $a$ the polynomial decay condition  
    \begin{align} \label{e:mix:app}
    \alpha(r) \le C r^{-a}, \qquad \textnormal{with} \qquad a>\frac{24 (M+2)}{M-3}
    \end{align}
    and that at the same time $N/T^\eta\to 0$ for 
    \begin{align} \label{e:eta:app}
    0 < \eta \le 2- \frac{2M}{a(M-2)+M}.
    \end{align}
If mixing coefficients are indeed exponentially decaying (as claimed in the main part of this paper), we may choose $a$ in the above conditions arbitrary large and thus the condition on $\eta$ is satisfied for any $\eta \in (0,2)$. Yet, the above formulation is interesting because it implies a more general trade-off between the strength of dependence and the relation of $N$ and $T$. More precisely, if dependence is stronger (polynomial mixing coefficients with smaller $a$), then $T$ needs to be larger compared to $N$. This makes intuitive sense, because more temporal data is needed to offset dependence. In our proofs, we will then refer by Assumption \ref{ass_2}, $iii)$ to the weaker condition \eqref{e:mix:app} and by Assumption \ref{ass_3}, $iii)$ to the condition \eqref{e:eta:app}.

\subsection{Proof of Lemma \ref{lem_1}}

\noindent We begin by considering the individual error
\[
E_i^{ind} := \E \big[ ( x_{i,T+1}'\hat\beta_i -y_{i,T+1})^2\big| \mathbf{X}\big].
\]
By definition of the OLS estimator and $y_{i,T+1}=x_{i,T+1}'\beta_i+\varepsilon_{i,T+1}$ it holds that
\begin{align*}
    x_{i,T+1}'\hat \beta_i -y_{i,T+1} = x_{i,T+1} ' [X_i' X_i]^{-1}X_i\varepsilon_i  - \varepsilon_{i,T+1}.
\end{align*}
Both terms are conditionally independent, which implies
\begin{align*}
    E_i^{ind} = & \E \bigg[Tr \Big[ \varepsilon_i' X_i  [X_i' X_i]^{-1}x_{i,T+1}x_{i,T+1}'[X_i' X_i]^{-1}X_i'\varepsilon_i  +\varepsilon_{i,T+1}'\varepsilon_{i,T+1}\Big]\bigg|\mathbf{X} \bigg]\\
    = & (\Sigma_N)_{i,i} Tr \Big[ \Sigma_T \Big(X_i  [X_i' X_i]^{-1}x_{i,T+1}x_{i,T+1}'[X_i' X_i]^{-1}X_i' \Big)\Big]+(\Sigma_N)_{i,i}(\Sigma_T)_{1,1}.
\end{align*}
In the second step we have used M3) and the fact that the covariance of the errors is separable (see Condition i) of Assumption \ref{ass_1}). We have thus shown the desired form of $E_i^{ind}$. In the next step, we consider the pooled errors
\[
E_i^{pool} := \E \big[ (x_{i,T+1}'\hat{\beta}^{pool} -y_{i,T+1})^2\big| \mathbf{X}\big].
\]
By definition of the pooled OLS estimator, we have
\begin{align*}
    & x_{i,T+1}'\hat \beta^{pool} -y_{i,T+1} =  x_{i,T+1}' \Big( \sum_{j=1}^N X_j' X_j \Big)^{-1}\sum_{j=1}^N \{ X_j' X_j\beta_j    + X_j'\varepsilon_j\} - x_{i,T+1}'\beta_i-\varepsilon_{i,T+1}\\
    = & x_{i,T+1}'\Big( \sum_{j=1}^N X_j' X_j \Big)^{-1} \sum_{j =1}^N \{X_j' X_j (\beta_j-\beta_i) + X_j'\varepsilon_j\}-\varepsilon_{i,T+1}.
\end{align*}
Hence
\begin{align*}
    & E_i^{pool} = \bigg(x_{i,T+1}'\Big( \sum_{j=1}^N X_j' X_j \Big)^{-1} \sum_{j =1}^N X_j' X_j (\beta_j-\beta_i)\bigg)^2\\
   & +Tr \bigg[ \Big( \sum_{j=1}^N X_j' X_j \Big)^{-1} x_{i,T+1} x_{i,T+1}'\Big( \sum_{j=1}^N X_j' X_j \Big)^{-1}\sum_{j,k}(\Sigma_N)_{j,k} X_j' \Sigma_T X_k\bigg]+(\Sigma_N)_{i,i}(\Sigma_T)_{1,1}.
\end{align*}

\subsection{Proof of Lemma \ref{lem_2}}
\noindent Notice that the terms $E_1, E_3$ are both $\mathbf{X}$-measurable and according to \eqref{e:lanXs} it suffices to show the desired rates in the lemma unconditionally.\\
We begin by investigating  $E_1$, defined in \eqref{e:E_1}. By some simple calculations, we get
\begin{align} \label{bound_E_1}
|E_1| \le  & \frac{1}{T}\max_{i=1,...,N}\Big\|(\Sigma_N)_{i,i}   [X_i' X_i/T]^{-1}x_{i,T+1}x_{i,T+1}'[X_i' X_i/T]^{-1}(X_i'\Sigma_T X_i)/T \Big\|_1 \\
\le &  \frac{C}{T} \max_{i=1,...,N}\|
[X_i' X_i/T]^{-1}\|_1 \|x_{i,T+1}x_{i,T+1}'\|_1 \|[X_i' X_i/T]^{-1}\|_1 \|(X_i'\Sigma_T X_i)/T \|_1. \nonumber
\end{align}
Here, we have used M4) in the first inequality and M5) in the second one. Using M1), M5) and finally Condition ii) of Assumption \ref{ass_2} shows that 
\[\|x_{i,T+1}x_{i,T+1}'\|_1 = \|x_{i,T+1}x_{i,T+1}'\|_2 \le \|x_{i,T+1}\|_2^2\le   c_3^2.
\]
Here $c_3>0$ is the constant from the named assumption.  Using Lemma \ref{lem:det:1} (parts i) and v) respectively) shows that 
\[
\max_i\|[X_i' X_i/T]^{-1}\|_1= \mathcal{O}_P(1), \qquad \max_i  \|(X_i'\Sigma_T X_i)/T \|_1= \mathcal{O}_P(1).
\]
It hence follows that $E_1=\mathcal{O}_P(1/T)$.\\
Next, we turn to $E_3$, defined in \eqref{e:E_3}. Using M4) and then M5) shows that
\begin{align*} \label{bound_E_3}
    |E_3| \le \frac{1}{NT} \Big(\Big\|\Big\{\frac{1}{NT}\sum_{j=1}^N X_j' X_j \Big\}^{-1} \Big\|_1^2\Big) \Big( \frac{1}{N}\sum_{i=1}^N \|x_{i,T+1}x_{i,T+1}' \|_1 \Big) \Big(\frac{1}{N}\sum_{j,k}|(\Sigma_N)_{j,k}|\cdot  \| (X_j' \Sigma_T X_k)/T\|_1\Big).
\end{align*}
All three factors in round brackets on the right side are of order $\mathcal{O}_P(1)$. For the first one this holds by Lemma \ref{lem:det:1}, part iii), for the second one by Assumption \ref{ass_2}, part ii) and for the last one by Lemma \ref{lem:det:1}, parts iv) and v). This demonstrates the desired rate 
\[
E_3 = \mathcal{O}_P \Big(\frac{1}{NT} \Big).
\]

\hfill \qed

\subsection{Proof of Theorem \ref{theo_main}:}
\begin{proof}
\noindent Recall the following version of the third binomial formula: $a^2-b^2=2(a-b)b+(a-b)^2$. We can use it to observe the following decomposition
$
 \sqrt{N}T(\hat E-E_2) = E_{2,1}+E_{2,2},
$
where 
\begin{align*}
     E_{2,1}:= &  \frac{2T}{ \sqrt{N}} \sum_{i=1}^N \bigg(x_{i,T+1}'\Big( \sum_{j=1}^N X_j' X_j \Big)^{-1} \sum_{j =1}^N X_j' X_j ([\hat \beta_j-\beta_j]-[\hat \beta_i-\beta_i])\bigg)'\\& \quad \qquad \bigg(x_{i,T+1}'\Big( \sum_{j=1}^N X_j' X_j \Big)^{-1} \sum_{j =1}^N X_j' X_j (\beta_j-\beta_i)\bigg)\\
     E_{2,2}:= & \frac{T}{ \sqrt{N}} \sum_{i=1}^N \bigg(x_{i,T+1}'\Big( \sum_{j=1}^N X_j' X_j \Big)^{-1} \sum_{j =1}^N X_j' X_j ([\hat \beta_j-\beta_j]-[\hat \beta_i-\beta_i])\bigg)^2.
\end{align*}
\noindent First, we turn to the  term $E_{2,2}$ that we rewrite as
\[
E_{2,2}=  \frac{T}{ \sqrt{N}} \sum_{i=1}^N \bigg(x_{i,T+1}'\Big( \sum_{j=1}^N X_j' X_j \Big)^{-1} \sum_{j =1}^N X_j' \varepsilon_j-x_{i,T+1}'(X_i'X_i)^{-1}X_i'\varepsilon_i\bigg)^2.
\]
Using the second binomial formula implies
$E_{2,2} = E_{2,2,1}-2E_{2,2,2}+E_{2,2,3}$, with
\begin{align*}
    E_{2,2,1}&=\frac{T}{ \sqrt{N}} \sum_{i=1}^N (x_{i,T+1}'(X_i'X_i)^{-1}X_i'\varepsilon_i)^2\\
    E_{2,2,2}&=\frac{1}{ \sqrt{N}} \bigg[\frac{1}{\sqrt{NT}} \sum_{j =1}^N \varepsilon_j'X_j \Big(\sum_{\ell=1}^N X_\ell' X_\ell /(NT)\Big)^{-1}\bigg]  \bigg[\frac{1}{\sqrt{NT}} \sum_{i=1}^N x_{i,T+1}x_{i,T+1}'(X_i'X_i/T)^{-1}X_i'\varepsilon_i\bigg]\\
    E_{2,2,3}&= \frac{1}{ \sqrt{N}} \bigg[\Big(\sum_{j=1}^N X_j' X_j /(NT)\Big)^{-1} \sum_{j =1}^N \frac{1}{\sqrt{NT}} X_j' \varepsilon_j \bigg]'\\
    & \qquad \times \Big(\frac{1}{N} \sum_{i=1}^Nx_{i,T+1}x_{i,T+1}'\Big) \bigg[\Big(\sum_{j=1}^N X_j' X_j/(NT) \Big)^{-1} \sum_{j =1}^N \frac{1}{\sqrt{NT}} X_j' \varepsilon_j\bigg].
\end{align*}
According to Lemma \ref{lem:help1} below, it follows that 
\[
E_{2,2,2}, E_{2,2,3} = \mathcal{O}_P^{|\mathbf{X}}(1/ \sqrt{N}).
\]
Consequently, we have
\begin{align} \label{e:Edec}
 \sqrt{N}T(\hat E-E_1-E_2) =  E_{2,1} + \{E_{2,2,1}-\sqrt{N}TE_1\}+\mathcal{O}_P^{|\mathbf{X}}(1/ \sqrt{N}),
\end{align}
and we now analyze the two non-negligible terms on the right. We will write these terms as sums of mixing random variables and demonstrate a moment condition for each variable. These steps are fundamental for the central limit theorem that we subsequently employ. We begin with
\[
    E_{2,2,1}-\sqrt{N}T E_1 = \frac{1}{ \sqrt{N}}\sum_{k=1}^N E_{2,2,1}^{(k)}, 
\]
where
\[
E_{2,2,1}^{(k)}:= \big\{(x_{i,T+1}'(X_i'X_i/T)^{-1}X_i'\varepsilon_i/\sqrt{T})^2-\mathbb{E}[(x_{i,T+1}'(X_i'X_i/T)^{-1}X_i'\varepsilon_i/\sqrt{T})^2|\mathbf{X}]\big\}.
\]
 Now, we study the conditional fourth moment of $E_{2,2,1}^{(k)}$. According to Lemma \ref{lem:det:1}, parts viii) we observe that
\begin{align} \label{e:bound1}
& \mathbb{E}\big[|  E_{2,2,1}^{(k)}|^4\big|\mathbf{X}\big]\le C \mathbb{E}\big[\big((x_{i,T+1}'(X_i'X_i/T)^{-1}X_i'\varepsilon_i/\sqrt{T})^2 \big)^4\big|\mathbf{X}\big]= \mathcal{O}_P(1). \nonumber
\end{align}
Next, we will analyze $E_{2,1}$.  We decompose further into  $E_{2,1}=E_{2,1,1}+E_{2,1,2}$, with 
\[
E_{2,1,\ell}:= \frac{1}{\sqrt{N}}\sum_{k =1}^N E_{2,1,\ell}^{(k)}
\]
where 
\begin{align*}
    E_{2,1,1}^{(k)} := & 2\frac{\varepsilon_k'X_k}{\sqrt{T}}  \bigg( \sum_{j=1}^N \frac{X_j' X_j}{NT} \Big)^{-1} \Big\{ \sum_{i=1}^N \frac{x_{i,T+1}x_{i,T+1}'}{N}\Big( \sum_{j=1}^N \frac{X_j' X_j}{NT} \Big)^{-1} \sum_{j =1}^N \frac{X_j' X_j (\beta_j-\beta_i)}{N\sqrt{T}}\Big\}\bigg)\\
    E_{2,1,2}^{(k)}:= & 2 \frac{\varepsilon_k' X_k}{\sqrt{T}}\Big(\frac{X_k'X_k}{T}\Big)^{-1} x_{k,T+1}x_{k,T+1}'\Big( \sum_{j=1}^N \frac{X_j' X_j}{NT} \Big)^{-1} \sum_{j =1}^N \frac{X_j' X_j (\beta_j-\beta_k)}{N \sqrt{T}}.
\end{align*}
It is easy to see that each term $E_{2,1,\ell}^{(k)}$ is conditionally centered. Next, we can study the conditional fourth moment of $E_{2,1,\ell}^{(k)}$. For purposes of illustration, we confine the proof to the term $E_{2,1,1}^{(k)}$ and omit the easier case $E_{2,1,2}^{(k)}$. First, notice that $|E_{2,1,1}^{(k)}|=|Tr[E_{2,1,1}^{(k)}]|$, since $E_{2,1,1}^{(k)}$ is just a real number. We can then use M4) and M5) to obtain
\begin{align*}
\mathbb{E}\big[| E_{2,1,1,}^{(k)}|^4 \big| \mathbf{X} \big]\le & 4\mathbb{E}\bigg[\Big\|\frac{\varepsilon_k'X_k}{\sqrt{T}} \Big\|_1^4\bigg| \mathbf{X} \bigg] \bigg\|\Big( \sum_{j=1}^N \frac{X_j' X_j}{NT} \Big)^{-1}\bigg\|_1^4 \bigg\|\sum_{i=1}^N \frac{x_{i,T+1}x_{i,T+1}'}{N}\bigg\|_1^4 \\
& \bigg\|\Big( \sum_{j=1}^N \frac{X_j' X_j}{NT} \Big)^{-1} \bigg\|_1^4\bigg\|\sum_{j =1}^N \frac{X_j' X_j }{NT}\bigg\|_1^4 \|\max_{i,j} \sqrt{T}(\beta_j-\beta_i)\|_2^4.
\end{align*}
Each of the norms on the right-hand side if of order $\mathcal{O}_P(1)$, uniformly over $k$, which follows by Lemma \ref{lem:det:1} part  vii) (first factor), part iii) (second and fourth factor), part vi) (fifth factor), Assumption \ref{ass_2}, Condition ii) (third factor) and Assumption \ref{ass_3} Condition ii) (last factor). We thus conclude
\begin{align}\label{e:bound2}
\max_k \mathbb{E}\big[ |E_{2,1}^{(k)}|^4\big|\mathbf{X}\big]\le C \max_k \mathbb{E}\big[ |E_{2,1,1}^{(k)}|^4\big|\mathbf{X}\big]\big]+C\max_k \mathbb{E}\big[ E_{2,1,2}^{(k)}|^4\big|\mathbf{X}\big] = \mathcal{O}_P(1).
\end{align}
We can now express
\begin{align} \label{e:tildeE}
\tilde E:= E_{2,1} + \{E_{2,2,1}-\sqrt{N}TE_1\} = \frac{1}{\sqrt{N}}\sum_{k=1}^N  \Big\{(E_{2,1,1}^{(k)}+E_{2,1,2}^{(k)})+E_{2,2,1}^{(k)}\Big\} ,
\end{align}
and notice that according to \eqref{e:Edec} we have
\[
\sqrt{N}T(\hat E-E_1-E_2) = \tilde E+o_P^{|\mathbf{X}}(1).
\]
Now, to study the weak convergence behavior of $\tilde E$, we first define its conditional variance 
\begin{align} \label{e:tau_N}
\tilde \tau_N^2 := \mathbb{V}ar(\tilde E|\mathbf{X})
\end{align}
and therewith the conditional distribution function
\[
F_{\tilde E}(z):= \mathbb{P}(\tilde E/\tilde \tau_N \le z| \mathbf{X})
\]
for $z \in \mathbb{R}$. 
 Now, we want to apply a Berry Esseen Theorem to  $F_{\tilde E}(z)$. More precisely, we consider the difference
 \[
 \Delta_N:=\mathbb{E}\big[\sup_z|F_{\tilde E}(z)-\Phi(z)|\big| \mathbf{X} \big]
 \]
and want to apply Theorem 2 from \cite{sunklodas:1984}, which rests on two conditions: Strong mixing of the random variables and uniformly bounded moments of all terms. First, the random variables $\varepsilon_n$ are strongly mixing across $n$, with mixing coefficients $\alpha(n) \le C n^{-a}$ and $a>12$. Hence, the first condition of the named theorem is satisfied. The second condition of the named theorem requires uniform boundedness of fourth moments of $E_{2,1}^{(k)}, E_{2,2,1}^{(k)}$. While the conditional moments are not uniformly bounded, we can define for a constant $L>0$ the event
\[
E(L):=\big\{\mathbb{E}[(E_{2,1}^{(k)})^4|\mathbf{X}]\le L,  \mathbb{E}[(E_{2,2,1}^{(k)})^4|\mathbf{X}]\le L\big\}.
\]
Hence, we can decompose $\Delta_N$ as follows
\[
\Delta_N \le \mathbb{E}\big[\sup_z|F_{\tilde E}(z)-\Phi(z)|\big| \mathbf{X} \big]\mathbb{I}_{E(L)}+\mathbb{I}_{E(L)^c}=: \Delta_N^{(1)}+\Delta_N^{(2)}.
\]
For any fixed $L>0$, Theorem 2 in \cite{sunklodas:1984} implies that
\[
\Delta_N^{(1)} \overset{\mathbb{P}}{\to}0
\]
and evaluating the double limit
\[
\lim_{L \to \infty} \limsup_N \Delta_N^{(2)}\overset{\mathbb{P}}{=}0
\]
yields $\Delta_N=o_P(1)$. To complete the proof of weak convergence, we still have to prove that 
\[
\breve \Delta_N:=\mathbb{E}\big[\sup_z|F_{\breve E}(z)-\Phi(z)|\big| \mathbf{X} \big] \to 0,
\]
where $\breve E$ is defined in 
\eqref{e:ebreve}. If Assumption \ref{ass_var} holds, this follows by standard arguments, since our above proof implies $\breve E = \tilde E +o_P(1)$ and since $\tilde \tau_N^2 = \tau_N^2+o_P(1)$, as we will show below. Hence, we   omit the details. \\
We now proceed to the last part of this proof: To show that
\[
\tilde \tau_N^2 = \tau_N^2+o_P(1)
\]
where $\tau_N^2$ is defined in \eqref{e:deftau}. It turns out that the proof of this fact requires some rather lengthy calculations.\\
 In view of \eqref{e:tildeE} and \eqref{e:tau_N}
 it holds that
 \[
 \tilde \tau_N^2 = \mathbb{V}ar(\tilde E|\mathbf{X}) = \frac{1}{N}\sum_{i,k} c_{i,k}^{(1)}+c_{i,k}^{(2)}+c_{i,k}^{(3)}
 \]
where
\begin{itemize}
    \item[i)] $c_{i,k}^{(1)}:=\mathbb{E}[E_{2,2,1}^{(i)}(E_{2,1,1}^{(k)}+E_{2,1,2}^{(k)})|\mathbf{X}]$
    \item[ii)] $c_{i,k}^{(2)}:=\mathbb{E}[(E_{2,1,1}^{(i)}+E_{2,1,2}^{(i)})(E_{2,1,1}^{(k)}+E_{2,1,2}^{(k)})|\mathbf{X}]$
    \item[iii)] $c_{i,k}^{(3)}:=\mathbb{E}[E_{2,2,1}^{(i)}E_{2,2,1}^{(k)}|\mathbf{X}]$.
\end{itemize}
We now analyze each of the $c_{i,k}^{(\ell)}$ for $\ell=1,2,3$ separately.

\textbf{Part i):} In order to study $c_{i,k}^{(1)} $, we further decompose it into 
\[
c_{i,k}^{(1)} =c_{i,k}^{(1,1)} +c_{i,k}^{(1,2)}:=\mathbb{E}[E_{2,2,1}^{(i)}E_{2,1,1}^{(k)}|\mathbf{X}]+\mathbb{E}[E_{2,2,1}^{(i)}E_{2,1,2}^{(k)}|\mathbf{X}].
\]
It can be shown that 
\begin{align} \label{e:cconv}
    \max_{i,k}\{c_{i,k}^{(1,1)},  c_{i,k}^{(1,2)}  \}\to 0.
\end{align}
Heuristically this is this the case since we consider the covariance of an (approximately) centered chi-squared distribution $E_{2,2,1}^{(i)}$ times an (approximately) normal distribution $E_{2,1,1}^{(k)}, E_{2,1,2}^{(k)}$ which is should be $\approx 0$ regardless of the correlation between these variables.\\
We will confine our proof to $c_{i,k}^{(1,1)}$ since the proof for the other term works analogously. 
Notice that we can write $X_i'\varepsilon_i=\sum_t \varepsilon_{i,t} x_{i,t}$. We now define 
\begin{align} \label{e:def_v_i}
v_i:= &(X_i'X_i/T)^{-1}x_{i,T+1},\\ V:=&2\Big( \sum_{j=1}^N \frac{X_j' X_j}{NT} \Big)^{-1} \Big\{ \sum_{i=1}^N \frac{x_{i,T+1}x_{i,T+1}'}{N}\Big( \sum_{j=1}^N \frac{X_j' X_j}{NT} \Big)^{-1} \sum_{j =1}^N \frac{X_j' X_j (\beta_j-\beta_i)}{N\sqrt{T}}\Big\}\bigg).\label{e:def_V}
\end{align}
Consequently, we have
\begin{align} \label{e:err_sum}
 c_{i,k}^{(1,1)}  :=&| \mathbb{E}[E_{2,2,1}^{(i)} E_{2,1,1}^{(k)}|\mathbf{X}]|\\
= &  \nonumber \big|\mathbb{E}\big[[(v_i'X_i'\varepsilon_i/\sqrt{T})^2-\mathbb{E}\{(v_i'X_i'\varepsilon_i/\sqrt{T})^2|\mathbf{X}\} ] (\varepsilon_k X_k'V/\sqrt{T})|\mathbf{X}\big]\big|\\
\le & \frac{1}{T^{3/2}}\sum_{q,r,s} |\mathbb{E}[\varepsilon_{i,q} (x_{i,q}' v_i) \varepsilon_{i,r} (x_{i,r}' v_i) \varepsilon_{k,s} (x_{k,s}' V)|\mathbf{X}]| \nonumber\\
\le &  \frac{1}{T^{3/2}} \max_{i,q} \{|x_{i,q}' v_i|\}^2 \max_{i,s}\{| (x_{i,s}' V)|\}\sum_{q,r,s} |\mathbb{E}[\varepsilon_{i,q} \varepsilon_{i,r} \varepsilon_{k,s}] |. \nonumber
\end{align}
Let us consider the maximum on right side first. According to Lemma \ref{lem:rates_spec} we have 
\[
\max_i \|v_i\| =\mathcal{O}_P(1), \quad \|V\|=\mathcal{O}_P(1)
\]
and hence 
\begin{align} \label{e:xmax}
\max_{i,q} \{|x_{i,q}' v_i|\}^2 \max_{i,s}\{| (x_{i,s}' V)|\} \le \mathcal{O}_P(1) (\max_{i,t} \|x_{i,t}\|)^3 = \mathcal{O}_P(T^{3/M})   
\end{align}
where we have used Lemma \ref{lem:rates_spec} again to get the rate in the last equality. Next, we turn to the sum over the error terms on the right side of \eqref{e:err_sum}. 
We can carve up the sum along various subsets of indices. In the following we consider the sum over the subset of indices where $q \le r \le s$ (other combinations work analogously). We can now for a sufficiently small $\zeta>3M/2$ (precisely specified later) decompose this set into 
\begin{align*}
    \mathcal{I}_{close} := &\{(q,r,s): q \le r \le s, |q-r|\le T^{1/4-\zeta}, |r-s|\le T^{1/4-\zeta}\},\\
    \mathcal{I}_{far} :=& \{(q,r,s): q \le r \le s\}\setminus \mathcal{I}_{close}. 
\end{align*}
Notice that $|\mathcal{I}_{close}|=\mathcal{O}(T^{3/2-2\zeta})$ by a simple combinatorial argument: There are $T$ possible choices for $q$ and for each $q$ there are at most $2 T^{1/4-\zeta}$ possible choices  for $r$ and $s$ giving at most $4 T \times T^{1/4-\zeta} \times T^{1/4-\zeta}$ combinations of indices. Consequently
\begin{align} \label{e:sumx1}
\frac{1}{T^{3/2}}\sum_{(q,r,s) \in \mathcal{I}_{close}} |\mathbb{E}[\varepsilon_{i,q} \varepsilon_{i,r} \varepsilon_{k,s} ]| = \mathcal{O}_P(T^{-2\zeta}), 
\end{align}
uniformly over $i, k$. Here we have also used that
\[
|\mathbb{E}[\varepsilon_{i,q} \varepsilon_{i,r} \varepsilon_{k,s} ]| \le \max_{i,q}\mathbb{E}|\varepsilon_{i,q}|^{3} \le C,
\]
where we have used Assumption \ref{ass_1}, part i) and the fact that $M\ge 3$.
On the other hand, considering the sum over $\mathcal{I}_{far}$ yields
\begin{align}\label{e:sumx2}
    \frac{1}{T^{3/2}}\sum_{(q,r,s) \in \mathcal{I}_{far}} |\mathbb{E}[\varepsilon_{i,q} \varepsilon_{i,r} \varepsilon_{k,s}] | \le C \frac{T^3}{T^{3/2}} \alpha(\lfloor T^{1/4-\zeta} \rfloor)^{(M-3)/M}=\mathcal{O}(T^{3/2-a(1/4-\zeta)(M-3)/M}). 
\end{align}
Here we have used the fact that in $\mathcal{I}_{far}$ there always exists at least one pair of indices with distance at least $T^{1/4-\zeta}$. So, for example, suppose that $|q-r|>T^{1/4-\zeta}$. Then we can use Lemma 3.11 in \cite{DehMikBook02}, to see that 
\begin{align} \label{e:twoin}
    |\mathbb{E}[\varepsilon_{i,q} \varepsilon_{i,r} \varepsilon_{k,s}] |\le \{\mathbb{E}|\varepsilon_{i,q} |^M\}^{1/M} \{\mathbb{E}|\varepsilon_{i,r} \varepsilon_{k,s}|^{M/2}\}^{2/M}\alpha(|q-r|)^{(M-3)/M}.
\end{align}
Now, putting the bounds \eqref{e:err_sum}, \eqref{e:sumx1} and \eqref{e:sumx2} together yields
\[
c_{i,k}^{(1,1)} = \mathcal{O}\big( T^{3/M-2\zeta}\big)+\mathcal{O}\big( T^{3/M+3/2-a(1/4-\zeta)(M-3)/M}\big).
\]
The right side converges to $0$ at some polynomial speed, if the two following conditions hold
\[
\frac{3}{M}<2\zeta \quad \textnormal{and} \quad \frac{3}{M}+\frac{3}{2}<a \Big(\frac{1}{4}-\zeta\Big)\frac{M-3}{M}
\]
Since $M>12$, we can choose $3/(2M)<\zeta<1/8$. The first condition in \eqref{e:twoin} follows since $3/(2M)<\zeta$. For the second one we use $\zeta<1/8$ and see that it is implied by Assumption \ref{ass_2}, condition iii) which implies 
\[
\frac{6+3M}{2M}<\frac{a}{8}\frac{M-3}{M} \quad \Leftrightarrow \quad \frac{8(6+3M)}{M-3}<a.
\]
This completes the proof of \eqref{e:cconv}.

\textbf{Part ii):} Next, we consider the covariance $ c_{i,k}^{(2)} = \mathbb{E}[E_{2,2,1}^{(i)}E_{2,2,1}^{(k)}]$. We can write this covariance explicitly as
\begin{align} \label{e:err_sum2}
\mathbb{E}[E_{2,2,1}^{(i)}E_{2,2,1}^{(k)}]=& \frac{1}{T^2}\mathbb{E}\bigg\{\Big(\sum_{t} \varepsilon_{i,t} (x_{i,t}'v_i)\Big)^2\Big(\sum_{t} \varepsilon_{k,t} (x_{k,t}'v_k)\Big)^2\\
&  -\mathbb{E}\Big[\Big(\sum_{t} \varepsilon_{i,t} (x_{i,t}'v_i)\Big)^2\Big| \mathbf{X}\Big]\mathbb{E}\Big[\Big(\sum_{t} \varepsilon_{k,t} (x_{k,t}'v_k)\Big)^2\Big| \mathbf{X}\Big]\Big| \textbf{X}\bigg\} \nonumber\\
= & \frac{1}{T^2}\sum_{s,t,q,r} \big\{\mathbb{E}[ \varepsilon_{i,s} \varepsilon_{i,t} \varepsilon_{k,q}\varepsilon_{k,r}]-\mathbb{E}[ \varepsilon_{i,s} \varepsilon_{i,t}]\mathbb{E}[\varepsilon_{k,q} \varepsilon_{k,r}] \big\} \nu_{i,k,s,t,q,r},\nonumber\\
\textnormal{where} \qquad &\nu_{i,k,s,t,q,r}:= \big\{ (x_{i,s}'v_i)(x_{i,t}'v_i)(x_{k,q}'v_k)(x_{k,r}'v_k)\big\} . \label{e:nu}
\end{align}
Notice that according to Lemma \ref{lem:rates_spec} it holds uniformly in all indices that
\begin{align}\label{e:nubound}
 |\nu_{i,k,s,t,q,r}|= \mathcal{O}_P(T^{4/J}).
\end{align}
We will now analyze the sum on the right of \eqref{e:err_sum2} and we will do that once more by carving up the set of indices. For the rest of this proof, we  introduce some terminology, to facilitate referencing different objects. First, we will call a sum of the form
\[
\sum \mathbb{E}[ \varepsilon_{i,s} \varepsilon_{i,t} \varepsilon_{k,q}\varepsilon_{k,r}]
\]
an "unseparated sum", while we call a sum of the form 
\[
\sum \mathbb{E}[ \varepsilon_{i,s} \varepsilon_{i,t} ]\mathbb{E}[ \varepsilon_{k,q}\varepsilon_{k,r}]
\]
a "separated sum". The terms unseparated and separated refer to the separation of the expectation of the four random variables into two smaller expectations. Second, we will parcel up the index set over all combinations of $s,t,q,r$ into ordered subsets, where e.g. $s < t < q < r$ or $s <  r <  t <  q$. It will make a difference in the analysis if in an ordering both pairs $s,t$ and $q,r$ are next to each other, such as in  
\[
\widehat{s < t} < \widehat{q < r}, \quad \widehat{q < r} < \widehat{t<s}
\] 
or whether this is not so, as in
\[
s < q < r< t, \quad q < r < t<s.
\]
We call the sets where the two pairs are together "adjacent orderings" and those where at least one pair is separated "non-adjacent orderings". Finally, it is clear that besides strict orderings an analysis of the index combinations also has to take account of one or more equalities in the indices, say  $s = t < q < r$  or $s = t = q < r$. The analysis for these index sets will not be considered here. We only notice that it works along similar lines as for the strictly ordered indices.\\
In the following, we will consider two strict orderings of indices: First, $s < t < q < r$ as an example of an adjacent ordering and second $s < r < t < q$ as an example of a non-adjacent ordering. We begin with $s < t < q < r$ and parcel up these indices into a number of subsets:
\begin{align*}
    \mathcal{I}_{\le, \le, \le }:=&\{s < t < q < r: |s-t|\le T^{1/4},|t-q|\le T^{1/4},|q-r|\le T^{1/4} \},\\
    \mathcal{I}_{>, \le, \le }:=&\{s < t < q < r: |s-t|> T^{1/4},|t-q|\le T^{1/4},|q-r|\le T^{1/4} \},\\
     \mathcal{I}_{\le, >, \le }:=&\{s < t < q < r: |s-t|\le T^{1/4},|t-q|> T^{1/4},|q-r|\le T^{1/4} \},\\
    \cdot&\\
    \mathcal{I}_{>, >, > }:=&\{s < t < q < r: |s-t|> T^{1/4},|t-q|> T^{1/4},|q-r|> T^{1/4} \}.
\end{align*}
There are three different types of index sets that we need to analyze: First, $\mathcal{I}_{\le, \le, \le }$ where all indices are relatively close to one another. Then, $\mathcal{I}_{\le, >, \le }$ where the pairs $s,t$ and $q,r$ are close to each other, but the pairs are far away from one another and finally any other index set, where there is a large distance between the elements of at least one pair $s,t$ or $q,r$. Let us begin by considering $\mathcal{I}_{\le, \le, \le }$. By a simple combinatorial argument, we see that $|\mathcal{I}_{\le, \le, \le }| \le C T^{7/4}$ ($T$ choices for $s$ and $2 \times T^{1/4}$ for $t,q,r$ respectively).
It then follows that
\[
 \frac{1}{T^2}\sum_{(s,t,q,r) \in \mathcal{I}_{\le, \le, \le }} \big\{\mathbb{E}[ \varepsilon_{i,s} \varepsilon_{i,t} \varepsilon_{k,q}\varepsilon_{k,r}]-\mathbb{E}[ \varepsilon_{i,s} \varepsilon_{i,t}]\mathbb{E}[\varepsilon_{k,q} \varepsilon_{k,r}] \big\} \nu_{i,k,s,t,q,r} =\mathcal{O}_P(T^{-1/4 + 4/M})=o(1).
\]
Here, we have used the bound on $\nu_{i,k,s,t,q,r} $ from Lemma \ref{lem:rates_spec} in the first equality. In the second equality we have exploited that $M \ge 16$. Next, let us turn to one of those sets where at least one pair $s,t$ or $q,r$ is separated by a distance of $T^{1/4}$. As an example we consider $\mathcal{I}_{>, \le, \le }$. We can exploit that $|s-t|>T^{1/4}$ so that a mixing inequality becomes applicable (again Lemma 3.11 in \cite{DehMikBook02}), that implies
\[
|\mathbb{E}[ \varepsilon_{i,s} \varepsilon_{i,t} \varepsilon_{k,q}\varepsilon_{k,r}]|,| \mathbb{E}[ \varepsilon_{i,s} \varepsilon_{i,t}]\mathbb{E}[\varepsilon_{k,q} \varepsilon_{k,r}]| \le C \alpha(T^{1/4})^{(M-4)/M}.
\]
By a simple calculation we observe that $|\mathcal{I}_{>, \le, \le }| \le C T^{1+1+1/4+1/4}=  C T^{5/2}$. Using this number, the above mixing bound and the bound \eqref{e:nubound} we obtain
\begin{align*}
    & \frac{1}{T^2}\sum_{(s,t,q,r) \in \mathcal{I}_{>, \le, \le }} \big\{\mathbb{E}[ \varepsilon_{i,s} \varepsilon_{i,t} \varepsilon_{k,q}\varepsilon_{k,r}]-\mathbb{E}[ \varepsilon_{i,s} \varepsilon_{i,t}]\mathbb{E}[\varepsilon_{k,q} \varepsilon_{k,r}] \big\} \nu_{i,k,s,t,q,r}\\
    = &\mathcal{O}_P \bigg(\frac{ T^{2.5}T^{4/M}}{T^2}\alpha(\lfloor T^{1/4} \rfloor )^{(M-4)/M} \bigg) = \mathcal{O}_P \bigg( T^{1/2+4/J-(a/4)(M-4)/M} \bigg).
\end{align*}
Exploiting $M \ge 16$, we see that the right side is polynomially decaying to $0$, if 
\[
\frac{3}{4}<\frac{a}{4} \frac{M-4}{M} \quad \Leftrightarrow \quad \frac{3M}{(M-4)}<a.
\]
We can now iterate these arguments for the analysis of all index sets except $ \mathcal{I}_{\le, >, \le }$ and observe that in all cases the sum is negligible if the mixing coefficients are sufficiently fast decaying. The strongest condition that we obtain is for the index set $  \mathcal{I}_{>, >, > }$ that has $\mathcal{O}(T^4)$ elements and gives us the condition
\[
\mathcal{O}_P \big(T^{2+4/J}\alpha(T^{1/4})^{M/(M-4)} \big)=\mathcal{O}_P \bigg( T^{2+4/J-(a/4)(M-4)/M} \bigg).
\]
Here, the right side decays poynomially if
\[
\frac{9}{4}<\frac{a}{4} \frac{M-4}{M} \quad \Leftrightarrow \quad \frac{9M}{M-4}<a
\]
which holds due to Assumption \ref{ass_2}, condition iii). 
 Finally, we turn to the index set $ \mathcal{I}_{\le, >, \le }$. Here, we can use Lemma 3.11 in \cite{DehMikBook02} to see that
\begin{align*}
    \mathbb{E}[ \varepsilon_{i,s} \varepsilon_{i,t} \varepsilon_{k,q}\varepsilon_{k,r}] = \mathbb{E}[ \varepsilon_{i,s}\varepsilon_{i,t}]\mathbb{E}[  \varepsilon_{k,q}\varepsilon_{k,r}] + \mathcal{O}_P(\alpha(\lfloor  T^{1/4}\rfloor )^{(M-4)/M} )=\mathcal{O}_P \bigg( T^{1/2+4/J-(a/4)(M-4)/M} \bigg).
\end{align*}
Using the fact that $ |\mathcal{I}_{\le, >, \le }| \le C T^{5/2}$, this implies that
\begin{align*}
    & \frac{1}{T^2}\sum_{(s,t,q,r) \in \mathcal{I}_{>, \le, \le }} \big\{\mathbb{E}[ \varepsilon_{i,s} \varepsilon_{i,t} \varepsilon_{k,q}\varepsilon_{k,r}]-\mathbb{E}[ \varepsilon_{i,s} \varepsilon_{i,t}]\mathbb{E}[\varepsilon_{k,q} \varepsilon_{k,r}] \big\} \nu_{i,k,s,t,q,r} \\
    = & \mathcal{O}(\alpha(\lfloor T^{1/4}\rfloor )^{(M-4)/M}T^{1/2+4/J} ),
\end{align*}
where as before the right side is going to $0$ at polynomial speed, using the mixing condition iii) of Assumption \ref{ass_2} and similar arguments as before. Notice the subtle difference here: For any index set besides $\mathcal{I}_{\le, >, \le }$ both separated and unseparated sum are individually vanishing. On $\mathcal{I}_{\le, >, \le }$ they both are not vanishing, but they are close to each other and hence asymptotically their difference vanishes. 
Our derivations up to this point show that 
\begin{align} \label{e:sum_vanish}
      \frac{1}{T^2}\sum_{s < t < q < r}  \big\{\mathbb{E}[ \varepsilon_{i,s} \varepsilon_{i,t} \varepsilon_{k,q}\varepsilon_{k,r}]-\mathbb{E}[ \varepsilon_{i,s} \varepsilon_{i,t}]\mathbb{E}[\varepsilon_{k,q} \varepsilon_{k,r}] \big\} \nu_{i,k,s,t,q,r}=o(1).
\end{align}
Similarly, we can show for any adjacent ordering of indices that the difference of separated and non-separated sums vanishes. \\
Now, we turn to the different case of a non-adjacent set of indices and for the purpose of illustration choose $s < r < t < q$. By making the same sorts of arguments as before (further decomposing the set of indices) we obtain that here, the separated sum is always of vanishing order  and thus
\begin{align} \label{e:nonad1}
    \frac{1}{T^2}\sum_{s < r < t < q}  |\mathbb{E}[ \varepsilon_{i,s} \varepsilon_{i,t}]\mathbb{E}[\varepsilon_{k,q} \varepsilon_{k,r}]  \nu_{i,k,s,t,q,r}|=o_P(1).
\end{align}
Roughly speaking, the reason is the following: Defining the index sets in analogy to before, the sum over $\mathcal{I}_{\le, \le, \le }$ is again negligible due to a small number of terms. For any other index set, there exist at least two indices $s,r$ or $r,t$ or $t,q$ that are further apart than $T^{1/4}$. For sake of argument, let us suppose $|t-q|>T^{1/4}$. Then, because $r<t<q$, it follows that $|r-q|>T^{1/4}$. Then, using Lemma 3.11 in \cite{DehMikBook02} we obtain the inequality
\[
|\mathbb{E}[\varepsilon_{k,q} \varepsilon_{k,r}]| \le C \alpha(\lfloor T^{1/4} \rfloor )^{(M-2)/M}. 
\]
Using this argument for all terms, we obtain that
\begin{align*}
& \frac{1}{T^2}\sum_{s < r < t < q}  |\mathbb{E}[ \varepsilon_{i,s} \varepsilon_{i,t}]\mathbb{E}[\varepsilon_{k,q} \varepsilon_{k,r}]  \nu_{i,k,s,t,q,r}|  \\
\le & \frac{1}{T^2}\sum_{s < r < t < q}  C \alpha(T^{1/4})^{(M-2)/M} |\nu_{i,k,s,t,q,r}| = \mathcal{O}_P\Big(T^{2+4/J}\alpha(\lfloor T^{1/4} \rfloor )^{(M-2)/M} \Big)\\
=& \mathcal{O}_P\Big(T^{2+4/J-a(M-2)/(4M)} \Big),
\end{align*}
where we have sued that the sum contains less than $T^4$ terms in total and that $|\nu_{i,k,s,t,q,r}|=\mathcal{O}_P(T^{4/J})$ according to Lemma \ref{lem:rates_spec}. As before, it follows that the right side is $o_P(1)$. This means that the separated sum will be asymptotically negligible. However, the non-separated sum is not negligible and we can show that
\begin{align} \label{e:nonad2}
\frac{1}{T^2}\sum_{s < r < t < q} \mathbb{E}[ \varepsilon_{i,s} \varepsilon_{i,t}\varepsilon_{k,q} \varepsilon_{k,r}]  \nu_{i,k,s,t,q,r} = \frac{1}{T^2}\sum_{s < r < t < q} \mathbb{E}[ \varepsilon_{i,s} \varepsilon_{k,r}]\mathbb{E}[\varepsilon_{i,t} \varepsilon_{k,q}]  \nu_{i,k,s,t,q,r} +o_P(1).  
\end{align}
Now, let us piece these results together. First, recall \eqref{e:err_sum2} and appreciate that we can decompose the right side further into 
\begin{align*}
    \frac{1}{T^2}\sum_{\substack{s,t,q,r \\ all \,indices\\ different}} \big\{\mathbb{E}[ \varepsilon_{i,s} \varepsilon_{i,t} \varepsilon_{k,q}\varepsilon_{k,r}]-\mathbb{E}[ \varepsilon_{i,s} \varepsilon_{i,t}]\mathbb{E}[\varepsilon_{k,q} \varepsilon_{k,r}] \big\} \nu_{i,k,s,t,q,r}\\
    + \frac{1}{T^2}\sum_{\substack{s,t,q,r \\ at \, least \, two \\ indices\, equal}} \big\{\mathbb{E}[ \varepsilon_{i,s} \varepsilon_{i,t} \varepsilon_{k,q}\varepsilon_{k,r}]-\mathbb{E}[ \varepsilon_{i,s} \varepsilon_{i,t}]\mathbb{E}[\varepsilon_{k,q} \varepsilon_{k,r}] \big\} \nu_{i,k,s,t,q,r} =: S_1+ S_2.
\end{align*}
Let us focus on $S_1$. The quadruple of  indices falls into one of the strictly ordered combinations considered above, adjacent or non-adjacent. In total there are $24$ possible orderings of four indices,  $8$ of them adjacent and $16$ non-adjacent. Our above proof shows the following three facts:\\
First, whenever we consider an adjacent ordering of indices it follows that (according to \eqref{e:sum_vanish})
\begin{align*}
      \frac{1}{T^2}\sum  \big\{\mathbb{E}[ \varepsilon_{i,s} \varepsilon_{i,t} \varepsilon_{k,q}\varepsilon_{k,r}]-\mathbb{E}[ \varepsilon_{i,s} \varepsilon_{i,t}]\mathbb{E}[\varepsilon_{k,q} \varepsilon_{k,r}] \big\} \nu_{i,k,s,t,q,r}=o(1).
\end{align*}
Second, whenever an ordering is non-adjacent, we have due to \eqref{e:nonad1} a vanishing separated sum
\[
\frac{1}{T^2}\sum  |\mathbb{E}[ \varepsilon_{i,s} \varepsilon_{i,t}]\mathbb{E}[\varepsilon_{k,q} \varepsilon_{k,r}]  \nu_{i,k,s,t,q,r}|=o_P(1).
\]
Third, in the non-adjacent case it follows that (according to \eqref{e:nonad2}) 
\[
\frac{1}{T^2}\sum \mathbb{E}[ \varepsilon_{i,s} \varepsilon_{i,t}\varepsilon_{k,q} \varepsilon_{k,r}]  \nu_{i,k,s,t,q,r} = \frac{1}{T^2}\sum \mathbb{E}[ \varepsilon_{i,s} \varepsilon_{k,r}]\mathbb{E}[\varepsilon_{k,q} \varepsilon_{i,t}]  \nu_{i,k,s,t,q,r} +o_P(1).
\]
We conclude, because of these observations that
\begin{align*}
    S_1 = &   \frac{16}{T^2}\sum_{s < r < t < q}  \mathbb{E}[ \varepsilon_{i,s} \varepsilon_{k,r}]\mathbb{E}[\varepsilon_{i,t} \varepsilon_{k,q}]  \nu_{i,k,s,t,q,r} +o_P(1) \\
    = &  \frac{16}{T^2}\sum_{s < r < t < q}  (\Sigma_T)_{s,r}(\Sigma_T)_{t,q} (\Sigma_N)_{i,k}^2 \nu_{i,k,s,t,q,r} + o_P(1)\\
    = &  2 \Big(\frac{8}{T^2}\sum_{s < t < q < r}  (\Sigma_T)_{s,r}(\Sigma_T)_{t,q} (\Sigma_N)_{i,k}^2 \nu_{i,k,s,t,q,r}\Big) +o_P(1) \\
    =&  \frac{2}{T^2}\sum_{\substack{s,t,q,r \\ all \,indices\\ different}} (\Sigma_T)_{s,r}(\Sigma_T)_{t,q} (\Sigma_N)_{i,k}^2 \nu_{i,k,s,t,q,r}+o_P(1).
\end{align*}
In the last step, we have used the fact that there are 8 adjacent ordering for the four indices and only on these orderings the separated sum is non-vanishing.
Making analogous arguments for the combinations where not all indices $s,t,q,r$ are different we can analyze $S_2$ and get that
\begin{align*}
    S_1+S_2 =  \frac{2(\Sigma_N)_{i,k}^2}{T^2}\sum_{s,t,q,r} (\Sigma_T)_{s,r}(\Sigma_T)_{t,q}  \nu_{i,k,s,t,q,r}.
\end{align*}
Now, plugging in the definition of $\nu_{i,k,s,t,q,r}$ in \eqref{e:nu}, we obtain that
\begin{align} \label{e:c3one}
c_{i,k}^{(3)} = \bar  c_{i,k}^{(3)} +\mathcal{O}_P(T^{-\zeta'}). 
\end{align}
Here, 
\begin{align}\label{e:c3two}
\bar  c_{i,k}^{(3)} := 2(\Sigma_N)_{i,k}^2 (x_{i,T+1}'(X_i'X_i/T)^{-1}[X_i'\Sigma_T X_k'/T](X_k'X_k/T)^{-1}x_{k,T+1})^2.
\end{align}
and the rate \refeq{e:c3one} is uniform in $i,k$ and holds for a sufficiently small $\zeta'>0$, essentially because all decay rates in this part of the proof were at some (slow) polynmial speed in $T$.

\textbf{Part ii):} Finally, we consider the covariance $\mathbb{E}(E_{2,1,1}^{(i)}+E_{2,1,2}^{(i)})(E_{2,1,1}^{(k)}+E_{2,1,2}^{(k)})$. To make our notation slightly more convenient, let us define 
\begin{align*}
    \Lambda := & \Big( \sum_{j=1}^N \frac{X_j' X_j}{NT} \Big)^{-1} \Big\{ \sum_{i=1}^N \frac{x_{i,T+1}x_{i,T+1}'}{N}\Big( \sum_{j=1}^N \frac{X_j' X_j}{NT} \Big)^{-1} \sum_{j =1}^N \frac{X_j' X_j (\beta_j-\beta_i)}{N\sqrt{T}}\Big\},\\
    \Lambda_k := & \Big(\frac{X_k'X_k}{T}\Big)^{-1} x_{k,T+1}x_{k,T+1}'\Big( \sum_{j=1}^N \frac{X_j' X_j}{NT} \Big)^{-1} \sum_{j =1}^N \frac{X_j' X_j (\beta_j-\beta_k)}{N \sqrt{T}}.
\end{align*}
We now obtain
\begin{align*}
   c_{i,k}^{(2)} = \mathbb{E}(E_{2,1,1}^{(i)}+E_{2,1,2}^{(i)})(E_{2,1,1}^{(k)}+E_{2,1,2}^{(k)}) = &  \frac{4}{N} ( \Lambda+ \Lambda_i)' \frac{X_i' (\Sigma_N)_{i,k} \Sigma_T X_k}{T}  ( \Lambda+ \Lambda_k).
\end{align*}
\textbf{Conclusion:} Remember that we want to show $ \tilde \tau_N^2=  \tau_N^2 + o_P(1)$ and that  
 \[
 \tilde \tau_N^2  = \frac{1}{N}\sum_{i,k} c_{i,k}^{(1)}+c_{i,k}^{(2)}+c_{i,k}^{(3)}, \qquad \tau_N^2  = \frac{1}{N}\sum_{i,k} c_{i,k}^{(2)}+\bar c_{i,k}^{(3)}.
 \]
 Hence, we get for the difference of the variances
 \begin{align*}
   | \tilde \tau_N^2- \tau_N^2| = &\frac{1}{N}\sum_{|i-k|<T^{\zeta'/2}}| c_{i,k}^{(3)}-\bar c_{i,k}^{(3)}|+   \frac{1}{N}\sum_{|i-k|>T^{\zeta'/2}}|c_{i,k}^{(3)}| +\frac{1}{N}\sum_{|i-k|>T^{\zeta'/2}} |\bar c_{i,k}^{(3)}|\\
   =: & S_{1,c}+S_{2,c}+S_{3,c}.
 \end{align*}
Using the rate \eqref{e:c3one} we obtain that 
\[
S_{1,c} = \mathcal{O}_P( T^{-\zeta'/2}).
\]
The analysis of $S_{2,c}, S_{3,c}$ works by the same method and therefore we focus on  $S_{2,c}$. Using the definition of $c_{i,k}^{(3)}:=\mathbb{E}[E_{2,2,1}^{(i)}E_{2,2,1}^{(k)}|\mathbf{X}]$ and the fact established earlier in this proof that $\mathbb{E}[  E_{2,2,1}^{(k)}|^4\big|\mathbf{X}]=\mathcal{O}_P(1)$ (uniformly in $k$), we obtain
\[
c_{i,k}^{(3)} \le \alpha(|i-k|)^{1/2}\mathbb{E}[  E_{2,2,1}^{(i)}|^4\big|\mathbf{X}]^{1/4}\mathbb{E}[  E_{2,2,1}^{(k)}|^4\big|\mathbf{X}]^{1/4} =\alpha(|i-k|)^{1/2} \mathcal{O}_P(1),
\]
where again the $\mathcal{O}_P$-term is uniform in $i,k$. Since by Assumption \ref{ass_2}, Condition iii)
\[
\alpha(|i-k|)^{1/2}  \le C (|i-k|+1)^{-a/2}
\]
is summable, we obtain by standard arguments that $S_{3,c}=o_P(1)$. Then $ \tilde \tau_N^2=  \tau_N^2 + o_P(1)$  follows, which completes the proof.
 
\end{proof}

\begin{lem} \label{lem:help1}
    Under the Assumptions of Lemma \ref{lem_3} it holds that
    \[
E_{2,2,2}, E_{2,2,3} = \mathcal{O}_P^{|\mathbf{X}}(1/(NT)).
\]
\end{lem}
\begin{proof}
In order to prove the lemma,
let us first investigate $E_{2,2,2}$. We decompose it into three parts
\[
E_{2,2,2} = \frac{1}{\sqrt{N}} E_{2,2,2,1}\cdot E_{2,2,2,2} \cdot E_{2,2,2,3},
\]
where we define
\begin{align*}
    E_{2,2,2,1} = & \frac{1}{\sqrt{NT}} \sum_{j =1}^N \varepsilon_j'X_j, \\
    E_{2,2,2,2} = & \Big(\sum_{\ell=1}^N X_\ell' X_\ell /(NT)\Big)^{-1},\\
   E_{2,2,2,3} = &  \frac{1}{\sqrt{NT}} \sum_{i=1}^N x_{i,T+1}x_{i,T+1}'(X_i'X_i/T)^{-1}X_i'\varepsilon_i.
\end{align*}
The desired rate is shown, if $E_{2,2,2,k}= \mathcal{O}_P^{\mathbf{X}}(1)$ for $k=1,2,3$. For $k=2$ this follows by Lemma \ref{lem:det:1} part iii). For $k=1,3$ it suffices to investigate $\mathbb{E}[\|E_{2,2,2,k}\|^2|\mathbf{X}]$, which follows by the conditional Chebychev inequality. Since both proofs are practically identical, we focus on $k=1$. $E_{2,2,2,1}$ is a vector of length $K$ and we can investigate any of its coordinates separately, say the $p$th one. Then, using Proposition \ref{moments_for_mixing} in the first step and Lemma \ref{lem:det:1}, part vii) in the second one yields
\[
\mathbb{E}[(E_{2,2,2,1})_p^2|\mathbf{X}]\le 
\frac{C}{N} \sum_{n=1}^N \mathbb{E}[\|\varepsilon_j'X_j/\sqrt{T}\|^4|\mathbf{X}]^{1/2} \le C \max_j \mathbb{E}[\|\varepsilon_j'X_j/\sqrt{T}\|^4|\mathbf{X}]^{1/2} = \mathcal{O}_P(1).
\]
The proof for $ E_{2,2,3}$ works analogously.

\end{proof}

\subsection{Proof of Lemma \ref{lem_4}}

We prove the following, stronger result than Lemma \ref{lem_4}:

\begin{lem}\label{lem_4_app}
    Suppose that Assumptions \ref{ass_1}, \ref{ass_2}, $i)-ii)$ and \ref{ass_3},  $ii)-iii)$ hold. Furthermore, suppose that \eqref{e:mix:app} and \eqref{e:eta:app} hold. Then, it follows that
    \[
    |\hat E_1-E_1|=\mathcal{O}^{|\mathbf{X}}_P\Big(\frac{b}{T^2}+\frac{1}{b^{-a(1-2/M)+1}}+\frac{ b^{7/4}}{\sqrt{N}T^{3/2}}\Big).
    \]
    If we choose $b=T^\rho$ with 
    \[
    \rho \in \bigg( \frac{\eta}{a(1-2/M)+1}, \min\Big(1- \frac{\eta}{2},\frac{2}{7}\Big)\bigg)
    \]
    the right side is of order $o^{|\mathbf{X}}_P(1/(\sqrt{N}T)$.
\end{lem}

\noindent We begin by defining the covariane matrix
\[
\Sigma^{(i,i)}:= \mathbb{E}[\varepsilon_i\varepsilon_i']
\]
and the
“ideal” estimator for the residual covariance of the $i$th individual with time lag $h$ as 
\begin{equation} \label{def_tilde_xi} 
    \tilde \xi_i(h) := \frac{(\varepsilon_i )_{1:T-h}' (\varepsilon_i )_{h+1:T}}{T-h-K} = \frac{1}{T-h-K} \sum_{i=1}^{T-h} \varepsilon_{i,t}\varepsilon_{i,t+h}.
\end{equation}
Therewith we define the “ideal estimated temporal matrix” as
$$
\tilde \Sigma^{(i,i)}(b) := (\tilde \xi_i(|s-t|) \mathbb{I}\{|s-t| <b\}))_{1\le s,t\le T },
$$
which corresponds to the matrix  $\hat \Sigma_{i,i} $ defined in \eqref{def_hat_Sigma_i}.
Let us also introduce the notation of $\Sigma^{(i,i)}(b)$, a banded version of $\Sigma^{(i,i)}$ with bandwidth $b$ and $\Sigma^{(i,i)}(b)^c$, which is defined by the identity $\Sigma^{(i,i)} = \Sigma^{(i,i)}(b)-\Sigma^{(i,i)}(b)^c$.  We can then decompose $\hat E_1-E_1 =  E_{1,1} +  E_{1,2}+E_{1,3}$, where
\begin{align*}
    E_{1,1} :=&\frac{1}{N}\sum_{i=1}^NTr\Big[
    \{\tilde \Sigma^{(i,i)}( b)-\Sigma^{(i,i)}( b)\} X_i  [X_i' X_i]^{-1}x_{i,T+1}x_{i,T+1}'[X_i' X_i]^{-1}X_i' \Big ] =: \frac{1}{N}\sum_{i=1}^N  E_{1,1}^{(i)},\\
    E_{1,2}:=&\frac{1}{N}\sum_{i=1}^NTr\Big[
    \{ \hat \Sigma^{(i,i)}( b)- \tilde \Sigma_i( b)\} X_i  [X_i' X_i]^{-1}x_{i,T+1}x_{i,T+1}'[X_i' X_i]^{-1}X_i' \Big ]=: \frac{1}{N}\sum_{i=1}^N  E_{1,2}^{(i)},\\
    E_{1,3}:=&\frac{1}{N}\sum_{i=1}^NTr\Big[
    \{ \Sigma^{(i,i)}( b)^c-\Sigma^{(i,i)}( b)\} X_i  [X_i' X_i]^{-1}x_{i,T+1}x_{i,T+1}'[X_i' X_i]^{-1}X_i' \Big ]=: \frac{1}{N}\sum_{i=1}^N  E_{1,3}^{(i)}.
\end{align*}
 We begin by analyzing the term $E_{1,3}$, and more precisely focus on $E_{1,3}^{(i)}$. We observe the upper bound
 \begin{align*}
     E_{1,3}^{(i)}\le &  \|\Sigma^{(i,i)}( b)^c-\Sigma^{(i,i)}( b) \|_\infty \|X_i  [X_i' X_i]^{-1}x_{i,T+1}x_{i,T+1}'[X_i' X_i]^{-1}X_i\|_1 \\
     = & \|\Sigma^{(i,i)}( b)^c-\Sigma^{(i,i)}( b) \|_\infty Tr[X_i  [X_i' X_i]^{-1}x_{i,T+1}x_{i,T+1}'[X_i' X_i]^{-1}X_i]\\
      \le & \frac{1}{T}\|\Sigma^{(i,i)}( b)^c-\Sigma^{(i,i)}( b) \|_\infty \|X_i'X_i/T\|_1
      \|x_{i,T+1}\|_1^2 \|[X_i' X_i/T]^{-1}\|_1^2
     =    \frac{\|\Sigma^{(i,i)}( b)^c-\Sigma^{(i,i)}( b)\|_\infty}{T}\mathcal{O}_P(1).
 \end{align*}
Here we have used M2) in the first, M4) in the second and M3), M4) and M2 in the third step. In the last equation, we have used  Lemma \ref{lem:det:1} parts i) and ix). Now, to upper bound the remaining term, we observe that according to M6)
\[
\|\Sigma^{(i,i)}( b)^c-\Sigma^{(i,i)}( b)\|_\infty\le \|\Sigma^{(i,i)}( b)^c-\Sigma^{(i,i)}( b)\|_{row} \le 2 \sum_{h \ge b} |\mathbb{E}[\varepsilon_{i,1}\varepsilon_{i,h} ]|.
\]
Using the mixing inequality (3.19) in \cite{DehMikBook02}, together with Assumptions \ref{ass_1} i) (moments) and Assumption \ref{ass_2} iii) (mixing) we see that
\[
|\mathbb{E}[\varepsilon_{i,1}\varepsilon_{i,h} ]| \le C \{\mathbb{E}\varepsilon_{i,1}^M\mathbb{E}\varepsilon_{i,h}^M\}^{1/M} (h-1)^{-a((M-2)/M)}.
\]
Consequently, we get
\[
\|\Sigma^{(i,i)}( b)^c-\Sigma^{(i,i)}( b)\|_\infty\le \|\Sigma^{(i,i)}( b)^c-\Sigma^{(i,i)}( b)\|_{row} \le C \sum_{h \ge b }(h-1)^{-a((M-2)/M)} \le C b^{-a((M-2)/M)+1}.
\]
Now, if we choose $b=C T^{2-\rho}$ and for $\rho \in [\eta,2)$, it is easy to see that 
$
b^{-a((M-2)/M)+1} \le C T^{-\eta} 
$ if 
\[
\rho \le 2 - \frac{\eta}{a(M-2)/M+1}.
\]
This condition can be fulfilled, if the above inequality holds with $\rho=\eta$, which leads by straightforward transformations to the condition
\[
\eta \le  2 \frac{a(M-2)/M+1}{a(M-2)/M+2} = 2 - \frac{M}{a(M-2)+2M}
\]
which holds according to Assumption \ref{ass_3} iii).\\
Next, we turn to the analysis of $E_{1,2}$. For this purpose, let us introduce the annihilation matrix $\mathbf{A}_{j:k}$, that equals an identity matrix in all rows, except the $j$th to the $k$th, which are set equal to $0$. Now, let us study an individual term $E_{1,2}^{(i)}$, which can be upper bounded by
\begin{align*}
&Tr\Big[X_i'\{ \hat \Sigma^{(i,i)}( b)- \tilde \Sigma^{(i,i)}( b)\} X_i  [X_i' X_i]^{-1}x_{i,T+1}x_{i,T+1}'[X_i' X_i]^{-1} \Big ]\\
\le &\|X_i'\{ \hat \Sigma^{(i,i)}( b)- \tilde \Sigma^{(i,i)}( b)\} X_i/T^2\|_2 \|[X_i' X_i/T]^{-1}\|_2^2 \|x_{i,T+1}\|_2^2 = \mathcal{O}_P(1) \|X_i' \{\hat \Sigma^{(i,i)}( b)- \tilde \Sigma^{(i,i)}( b)\} X_i/T^2\|_2 \\
\le &\mathcal{O}_P(1) \|X_i' \{\hat \Sigma^{(i,i)}( b)- \tilde \Sigma^{(i,i)}( b)\} X_i/T^2\|_\infty \le \mathcal{O}_P(1) \|X_i\|_\infty^2/T \| \{\hat \Sigma^{(i,i)}( b)- \tilde \Sigma^{(i,i)}( b)\}/T \|_\infty  \\
\le& \mathcal{O}_P(1) \| \{\hat \Sigma^{(i,i)}( b)- \tilde \Sigma^{(i,i)}( b)\} /T\|_{\infty}
\end{align*}
In the first step we have used  M5) and in the second one Lemma \ref{lem:det:1} part ix), together with the boundedness of the $x_{i,T+1}$ in Assumption \ref{ass_2}, Condition ii).  In the third step, we have used that on the space of $K \times K$ matrices all matrix norms are equivalent and in the fourth step M5) once more. In the last step, we have used that the Frobenius norm upper bounds the spectral norm, leading to the inequalities
\[
\|X_i\|_\infty^2/T \le \|X_i\|_2^2/T = \|X_i'X_i/T\|_2=\mathcal{O}_P(1),
\]
where we have used Lemma \ref{lem:det:1} part i) to get a uniform $\mathcal{O}$-term over $i$ on the right side. Now, recall that according to M6) it holds that 
\begin{align*}
&\| \{\hat \Sigma^{(i,i)}( b)- \tilde \Sigma^{(i,i)}( b)\} /T\|_{\infty} \le \| \{\hat \Sigma^{(i,i)}( b)- \tilde \Sigma^{(i,i)}( b)\} /T\|_{row} \\
\le &  \frac{2}{T} \sum_{h=1}^b |(\hat \Sigma^{(i,i)})_{h,1}-(\tilde \Sigma^{(i,i)})_{h,1}|.
\end{align*}
Now, let us consider for some fixed $h$ one of the terms on the right. Notice, that we can rewrite it, using the annihilator matrices $\mathbf{A}$ as
\begin{align*}
&|\hat \varepsilon_i' \mathbf{A}_{1:h}\mathbf{A}_{T-h+1:T}\hat \varepsilon_i-  \varepsilon_i' \mathbf{A}_{1:h}\mathbf{A}_{T-h+1:T} \varepsilon_i|/T\\
= & \big|Tr\big[ \{\hat \varepsilon_i\hat \varepsilon_i' -  \varepsilon_i\varepsilon_i' \}\mathbf{A}_{1:h}\mathbf{A}_{T-h+1:T} \big]\big|/T\\
\le& \big|Tr\big[ \{(\hat \varepsilon_i-\varepsilon_i)\hat \varepsilon_i' \}\mathbf{A}_{1:h}\mathbf{A}_{T-h+1:T} \big]\big|/T+\big|Tr\big[ \{ \varepsilon_i (\hat \varepsilon_i-\varepsilon_i)'\}\mathbf{A}_{1:h}\mathbf{A}_{T-h+1:T} \big]\big|/T.
\end{align*}
Both terms can be analyzed similarly and hence we focus on the (more difficult) first one. Recall that $ \varepsilon_i-\hat\varepsilon_i = X_i(X_i'X_i)^{-1}X_i'\varepsilon_i$, where $X_i(X_i'X_i)^{-1}X_i'$ is a projection matrix on a space of dimension $K$. Consequently, we have
\begin{align*}
   & \big|Tr\big[ \{(\hat \varepsilon_i-\varepsilon_i)\hat \varepsilon_i' \}\mathbf{A}_{1:h}\mathbf{A}_{T-h+1:T} \big]\big| =\big| Tr\big[ X_i(X_i'X_i)^{-1}X_i'\varepsilon_i \varepsilon_i'X_i'(X_i'X_i)^{-1}X_i\mathbf{A}_{1:h}\mathbf{A}_{T-h+1:T}\big]\big|\\
   \le & \|X_i(X_i'X_i)^{-1}X_i'\varepsilon_i \varepsilon_i'X_i'(X_i'X_i)^{-1}X_i\|_1 \| \mathbf{A}_{1:h}\mathbf{A}_{T-h+1:T}\|_\infty \\
   = & Tr[X_i(X_i'X_i)^{-1}X_i'\varepsilon_i \varepsilon_i'X_i'(X_i'X_i)^{-1}X_i]=Tr[\varepsilon_i \varepsilon_i'X_i'(X_i'X_i)^{-1}X_i]
\end{align*}
Here we have used in the first inequality M4) and M2), subsequently, then that $\| \mathbf{A}_{1:h}\mathbf{A}_{T-h+1:T}\|_\infty=1$, M4) again (the case of equality) and finally that M3), togeter with the fact that $X_i'(X_i'X_i)^{-1}X_i$ is a projection.
So, up to this point we have proved that
\[
E_{1,2} \le \mathcal{O}_P(1) \frac{b}{T^2} \Big\{\frac{1}{N} \sum_{i=1}^NTr\big[ \varepsilon_i'\varepsilon_i X_i(X_i'X_i)^{-1}X_i'\big]\Big\}.
\]
The last factor is non-negative and we observe that its conditional expectation satisfies
\begin{align*}
    & \mathbb{E}\Bigg[ \Big\{\frac{1}{N} \sum_{i=1}^NTr\big[ \varepsilon_i'\varepsilon_i X_i(X_i'X_i)^{-1}X_i'\big]\Big\}\bigg| \mathbf{X} \bigg]  = 
    \frac{1}{N} \sum_{i=1}^N (\Sigma_N)_{i,i}Tr\big[  \Sigma_T X_i(X_i'X_i)^{-1}X_i'\big]\\
    \le & \frac{1}{N} \sum_{i=1}^N (\Sigma_N)_{i,i} \|\Sigma_T\|_\infty  \|X_i(X_i'X_i)^{-1}X_i'\|_1 \le \frac{1}{N} \sum_{i=1}^N (\Sigma_N)_{i,i} \|\Sigma_T\|_\infty  Tr [X_i(X_i'X_i)^{-1}X_i'] \\
    \le &    \frac{1}{N} \sum_{i=1}^N (\Sigma_N)_{i,i} \|\Sigma_T\|_\infty  Tr [Id_{K \times K}] \le C.
\end{align*}
Here, we have used the covariance structure from Assumption \ref{ass_1}, part i) in the first step, M4) and M2) in the second step, M4) in the case of equality in the third step, M3) in the fourth step and finally the fact hat $\|\Sigma_T\|_\infty \le C$, from Lemma \ref{lem:det:1} part iv).
Combined, our considerations show that
\[
E_{1,2}= \mathcal{O}_P\Big(\frac{b}{T^2} \Big).
\]
Finally, we turn to $E_{1,1}$. Here, it is sufficient to show that
\[
\mathbb{E}[(E_{1,1})^2|\mathbf{X}]=o_P(1/(NT^2)).
\]
We employ Proposition \ref{moments_for_mixing} with $\phi=\xi=2$ and $d=1$, which yields
\begin{align} \label{e:E11bound}
\mathbb{E}[(E_{1,1})^2|\mathbf{X}]\le \frac{1}{N^2} \sum_{i=1}^N \{\mathbb{E}[(E_{1,1}^{(i)})^4|\mathbf{X}]\}^{1/2}.
\end{align}
The mixing condition of the named proposition is satisfied since $a>2$ in Assumption \ref{ass_2}, part iii). Now, it suffices to show that 
\[
\mathbb{E}[(E_{1,1}^{(i)})^4|\mathbf{X}] = \mathcal{O}_P(1/T^4).
\]
We observe the following bounds
\begin{align*}
    |E_{1,1}^{(i)}| = & \Big|Tr\Big[
    \{ \tilde \Sigma^{(i,i)}( b)-\Sigma^{(i,i)}( b)\} X_i  [X_i' X_i]^{-1}x_{i,T+1}x_{i,T+1}'[X_i' X_i]^{-1}X_i' \Big ]\Big| \\
    \le & \|X_i'\{\tilde \Sigma^{(i,i)}( b)-\Sigma^{(i,i)}( b)\} X_i/T^2\|_1\|x_{i,T+1}x_{i,T+1}'\|_1 \|[X_i' X_i/T]^{-1}\|_1^2\\
    \le & \mathcal{O}_P(1)\|X_i'\{\tilde \Sigma^{(i,i)}( b)-\Sigma^{(i,i)}( b)\} X_i/T^2\|_\infty.
\end{align*}
In the second step, we have used first M3), then M4) and then M5). In the third step we have used Lemma \ref{lem:det:1}  part ix) and Condition ii) of Assumption \ref{ass_2}. Finally, we have replaced the trace norm, by the spectral norm, since all matrix norms in $\mathbb{R}^K$ are equivalent. Using M5) shows that 
\begin{align*}
   & \|X_i'\{\tilde \Sigma^{(i,i)}( b)-\Sigma^{(i,i)}( b)\} X_i/T^2\|_\infty \le \|X_i'\|_\infty^2/T \|\tilde \Sigma^{(i,i)}( b)-\Sigma^{(i,i)}( b) T\|_\infty 
   \\
   \le & \|X_i'X_i/T\|_2 \|\{\tilde \Sigma^{(i,i)}( b)-\Sigma^{(i,i)}( b)\}/ T\|_\infty 
    = \mathcal{O}_P(1)\|\{\tilde \Sigma^{(i,i)}( b)-\Sigma^{(i,i)}( b)\}/ T\|_\infty.
\end{align*}
Here, we have used the bound M1) in the second step. We can thus focus on $\|\{\tilde \Sigma^{(i,i)}( b)-\Sigma^{(i,i)}( b)\}/ T\|_\infty$.
Now, by M6), we have 
\begin{align*}
& \|\{\tilde \Sigma^{(i,i)}( b)-\Sigma^{(i,i)}( b)\}/ T\|_\infty \le \|\{\tilde \Sigma^{(i,i)}( b)-\Sigma^{(i,i)}( b)\}/ T\|_{row} \\
\le & \frac{2}{T} \sum_{h=1}^b |Tr[\varepsilon_i'\varepsilon_i\mathbf{A}_{1:h}\mathbf{A}_{T-h+1:T}/(T-h)]-(\Sigma_N)_{i,i}(\Sigma_{T})_{1,h}]|.
\end{align*}
It is a special case of Hölder's inequality that
\[
\Big(\sum_{i=1}^b |x_i|\Big)^p \le b^{p-1} \sum_{i=1}^b |x_i|^p
\]
for $p>1$ and we can employ this bound to see that
\[
\|\{\tilde \Sigma^{(i,i)}( b)-\Sigma^{(i,i)}( b)\}/ T\|_\infty^4 \le  \frac{Cb^3}{T^4} \sum_{h=1}^b |Tr[\varepsilon_i'\varepsilon_i\mathbf{A}_{1:h}\mathbf{A}_{T-h+1:T}/(T-h)]-(\Sigma_N)_{i,i}(\Sigma_{T})_{1,h}]|^4.
\]
So, up to this point we have demonstrated that
\[
\mathbb{E}[\{E_{1,1}^{(i)}\}^{4}|\mathbf{X}] \le \mathcal{O}_P(1) \frac{b^{3}}{T^{4}}\mathbb{E}\Big[\sum_{h=1}^b |Tr[\varepsilon_i'\varepsilon_i\mathbf{A}_{1:h}\mathbf{A}_{T-h+1:T}/(T-h)]-(\Sigma_N)_{i,i}(\Sigma_{T})_{1,h}]|^{4}\Big|\mathbf{X}\Big].
\]
We now study for some fixed but arbitrary $h$ a term in the sum on the right side. We can express it as
\[
\mathbb{E}\Big[\Big(\frac{1}{T-h}\sum_{\ell=1}^{T-h}\varepsilon_{i,\ell}\varepsilon_{i,\ell+h}-\mathbb{E}[\varepsilon_{i,\ell}\varepsilon_{i,\ell+h}]\Big)^{4}\Big] \le \frac{C b^{4}}{T^{2}},
\]
where the constant $C$ can be chosen independently of $h$, if $b/T \to 0$. Here we have used the main theorem in \cite{fazekas:2000}. 
Combining this result with \eqref{e:E11bound}, shows that
\[
\{\mathbb{E}[(E_{1,1})^2|\mathbf{X}]\}^{1/2} \le \mathcal{O}_P\bigg(\frac{ b^{7/4}}{\sqrt{N}T^{3/2}}\bigg),
\]
completing the proof.

\newpage

\section{Additional details}

\begin{lem} \label{lem:det:1}
    Under Assumptions \ref{ass_1} and \ref{ass_2} it holds for any matrix norm $\|\cdot\|$ that
    \begin{itemize}
    \item[i)] 
    \[
    \max_i\|X_i'X_i/T\|=\mathcal{O}_P(1).
    \]
    \item[ii)] With probability tending to $1$ it holds that
    \[
    \min_i \lambda_{min}(X_i'X_i/T)\ge c_1/2.
\]
        \item[iii)]\[
    \Big\|\Big( \frac{1}{NT}\sum_{j=1}^N X_j' X_j \Big)^{-1} \Big\| = \mathcal{O}_P(1).
\]
\item[iv)]
\[
    \|\Sigma_T\|_\infty, \frac{1}{N}\sum_{j,k} |(\Sigma_N)_{j,k}| =\mathcal{O}(1). 
\]
    \item [v)] \[
        \max_{j,k}\| (X_j' \Sigma_T X_k)/T\| = \mathcal{O}_P(1).
    \]
        \item [vi)] \[
       \Big\| \frac{1}{NT}\sum_{j=1}^N X_j' X_j  \Big\|= \mathcal{O}_P(1).
    \]
    \item[vii)] 
    \[
    \max_i\mathbb{E}\bigg[\Big\|\frac{\varepsilon_i'X_i}{\sqrt{T}} \Big\|^8\bigg| \mathbf{X}\bigg] = \mathcal{O}_P(1).
    \]
    \item[viii)]
    \[
    \mathbb{E}\big[\big(\sqrt{T}x_{i,T+1}'(X_i'X_i)^{-1}X_i'\varepsilon_i\big)^8\big|\mathbf{X}\big] = \mathcal{O}_P(1).
    \]
        \item[ix)]
     \[
    \max_i\|(X_i'X_i/T)^{-1}\|=\mathcal{O}_P(1).
    \]
    \end{itemize}
\end{lem}

\begin{proof}
To prove part i), we simply notice that
\begin{align*}
     \max_i\|X_i'X_i/T\| \le  \max_i\|Q_i\| +  \max_i\|X_i'X_i/T-Q_i\|. 
\end{align*}
Assumption \ref{ass_2}, Condition i) implies that the terms on the right are $\mathcal{O}(1)$ and and $\mathcal{O}_P(1)$ respectively. \\
Part ii) follows by Weyl's inequality, which entails
\[
\lambda_{min}(X_i'X_i/T) \ge \lambda_{min}(Q_i)+\lambda_{min}(X_i'X_i/T-Q_i) \ge c_1 +o_P(1).
\]
Here we have used in the last step that $\min_i\lambda_{min}(Q_i)\ge c_1>0$ and that $\max_i \|X_i'X_i/T-Q_i\|=o_P(1)$ (both according to Assumption \ref{ass_2}, Condition i).\\
    In order to prove iii) we notice that for a $K\times K$ matrix of full rank $A$, it holds that 
    \[
    \|A^{-1}\|_2 \le Tr[A^{-1}] \le K \lambda_{max}(A^{-1}).
    \]
    The largest eigenvalue of $A^{-1}$ is $(\lambda_{min}(A))^{-1}$. This yields 
    \[
    \|A^{-1}\|_2 \le C(\lambda_{min}(A))^{-1}.
    \]
    Finally suppose that for $N \in \mathbb{N}$ $A$ can be expressed as the average $A= \frac{1}{N} \sum_{n=1}^N A_n$, where $A_n $ is a $K\times K$ matrix, with $\lambda_{min}(A)\ge c$, for some fixed $c>0$. Then it holds that 
    $\lambda_{min}(A)\ge \min_n \lambda_{min}(A_n)$ and
    \[
    \|A^{-1}\|_2 \le C (\min_n \lambda_{min}(A_n))^{-1}.
    \]
    Applying this result to $A_n=X_n'X_n/T$ and using the lower bound from Condition i) of Assumption \ref{ass_2} yields the desired result.\\
   To prove iv), we first recall that $((\Sigma_N)_{1,1})\Sigma_T = \mathbb{E}\varepsilon_1\varepsilon_1' $. Let us hence focus on $(\Sigma_N)_{1,1}\Sigma_T$. 
    Since according to M6) we can upper bound the spectral norm of a covariance matrix by its absolute row sum norm $\|\cdot\|_{row}$, it holds that 
    \[
    \|(\Sigma_N)_{1,1}\Sigma_T \|_\infty \le \sum_{t\ge 1  } 2|\mathbb{E}\varepsilon_{1,1} \varepsilon_{1,t} | \le C \sum_{t \ge 1}\{\mathbb{E}|\varepsilon_{1,1} |^M\}^{2/M}\alpha(t)^{1-2/M} \le C \sum_{t \ge 1} t^{-a(1-2/M)}.
    \]\\
    To obtain the second inequality, we  have used eq. (3.19) in \cite{DehMikBook02}.
    Since $a(1-2/M)<-1$ by assumption \ref{ass_2} conditon iii), the right side converges and the first claim follows.
     For the second claim, we use the fact that $(\Sigma_T)_{1,1}\Sigma_N$ is the covariance of the vector $(\varepsilon_{1,1},...,\varepsilon_{N,1})$. Using once more the mixing Condition iii) of Assumption \ref{ass_2}, we obtain 
     \[
     |\mathbb{E}\varepsilon_{1,j}\varepsilon_{1,i}| \le  \{\mathbb{E}|\varepsilon_{1,1}|^M \}^{2/M}\alpha(|i-j|+1)^{(M-2)/M}.
     \]
     Now, the absolute sum over all elements in $(\Sigma_T)_{1,1}\Sigma_N$ divided by $N$ is upper bounded by
     \[
     \sum_{i=1}^N \frac{1}{N} \sum_{j=1}^\infty |\mathbb{E}\varepsilon_{1,1}\varepsilon_{1,j}| \le 2 |\mathbb{E}\varepsilon_{1,1}\varepsilon_{1,j}| \le \sum_{j \ge 1}C\alpha(j)^{(M-2)/M}<\infty,
     \]
     because as before the right sum converges due to Condition iii) in \ref{ass_2}.\\
    Turning to part v) of the Lemma, we recall that for any matrix $A$, $\|A\|_2=\sqrt{Tr[A'A]}$, by definition of the Frobenius norm. Then using M3) in the first step and M4), M2), M5) in the second one implies 
\begin{align*}
    \| (X_j' \Sigma_T X_k)/T\|_2 = \sqrt{Tr[\{(X_kX_k')/T \}\{(X_jX_j')/T\} \Sigma_T \Sigma_T  ]} \le \sqrt{\|(X_k X_k')/T \|_1 \|(X_jX_j')/T\|_1 \|\Sigma_T\|_\infty^2}.
\end{align*}
We thus obtain
\[
\sqrt{\|\Sigma_T \Sigma_T\|_\infty} \max_{j,k} \| (X_j' \Sigma_T X_k)/T\|_2 \le \max_{j,k}\sqrt{\|(X_k X_k')/T \|_1 \|(X_jX_j')/T\|_1 }=\mathcal{O}_P(1).
\]
In the last step we have used parts i) and iv) of this lemma to bound the second factor. For the first factor, we can use that $\sqrt{\|\Sigma_T \Sigma_T\|_\infty} \le \|\Sigma_T \|_\infty \le \|\Sigma_T \|_{row}$. Here we have used M5) in the first and M6) in the second inequality. We have already established in part iv) of this lemma  that $\|\Sigma_T \|_\infty$ is bounded. \\
Part vi) of the Lemma follows by similar techniques as before and is therefore omitted.\\
To see part vii), we first rewrite the eighth moment as   
$$
\mathbb{E}\Big[\Big\|\frac{\varepsilon_k'X_k}{\sqrt{T}} \Big\|^8\Big|\mathbf{X}\Big]=\mathbb{E}\Big[\Big\| \frac{1}{\sqrt{T}}\sum_{t=1}^T\varepsilon_{k,t} x_{k,t}\Big\|^8\Big|\mathbf{X}\Big] \le C \max_{\ell=1,...,K} \mathbb{E}\Big[\Big| \frac{1}{\sqrt{T}}\sum_{t=1}^T\varepsilon_{k,t} (x_{k,t})_\ell\Big|^8\Big|\mathbf{X}\Big].
$$
Now, we employ Theorem 3 in \cite{Yoshihara1978}  which is a moment inequality for sums of mixing random variables. We use the parameter choices (in the terms of that theorem) of $m=8$, $\delta=M-8$, $a_t = (x_{k,t})_\ell$. 
The conditions of named Theorem apply since $\max_{k,t} \mathbb{E}(\varepsilon_{kt})^8 \le C$ and 
\[
\sum_{i \ge 1} (i+t)^{3} \alpha(i)^{(M-8)/M} \le C \sum_{i \ge 1} (i+t)^{3} (i+1)^{a(M-8)/M}<\infty.
\]
Now, applying named theorem for each $\ell$ then yields collectively
$$
\mathbb{E}\Big\| \frac{1}{\sqrt{T}}\sum_{t=1}^T\varepsilon_{k,t} x_{k,t}\Big\|^8 
\le \frac{C }{T^4} \bigg\{\Big(\sum_{t=1}^T \|x_{k,t}\|^2\Big)^4=CTr[X_kX_k'/T]^4 =CTr[X_k'X_k/T]^4 = \mathcal{O}_P(1).$$
Next, we prove part viii). For this purpose, we notice that
\begin{align*}
|\sqrt{T}x_{i,T+1}'(X_i'X_i)^{-1}X_i'\varepsilon_i |\le \|x_{i,T+1}'\| \|(X_i'X_i/T)^{-1}\| \|X_i'\varepsilon_i /\sqrt{T}\| = \mathcal{O}_P(1) \|X_i'\varepsilon_i /\sqrt{T}\|.
\end{align*}
Here we have used M5) in the first step and part i) of this Lemma, s.t. that $\mathcal{O}_P$-term on the right side is uniform in $i$ (notice that it is also measurable w.r.t. $\mathbf{X}$). Hence, it suffices to show that
\[
\mathbb{E}[\|X_i'\varepsilon_i /\sqrt{T}\|^8|\mathbf{X}]= \mathcal{O}_P(1).
\]
This follows directly from part vii) of this Lemma, concluding the proof.\\
Part ix) follows directly from part ii), since $X_i'X_i/T$ is a matrix of (fixed) dimension $K \times K$. Consequently
\[
\|(X_i'X_i/T)^{-1}\|_1 \le K \lambda_{max}((X_i'X_i/T)^{-1}) \le K \lambda_{min}(X_i'X_i/T)^{-1} = \mathcal{O}_P(1).
\]

\end{proof}

\begin{lem} \label{lem:rates_spec}
    Under the Assumptions of Theorem \ref{theo_main} and with the definitions of $v_i,V$ in \eqref{e:def_v_i} and \eqref{e:def_V} respectively, it holds that
    \begin{align*}
& \max_i \|v_i\| =\mathcal{O}_P(1), \quad \|V\|=\mathcal{O}_P(1)
,\quad \max_{i,t} \|x_{i,t}\| = \mathcal{O}_P(T^{1/J}),\\
& \max_{i,k,s,t,q,r}|\nu_{i,k,s,t,q,r}|= \mathcal{O}_P(T^{4/J}).
    \end{align*}
\end{lem}

\begin{prop} \label{moments_for_mixing}
Suppose that $(\Xi_z)_{z \in \mathcal{M}}$ is an $\alpha$-mixing field of random variables, that $\phi > 1$ and $\chi>0$ such that $\E|\Xi_z|^{\phi+\chi}<\infty \, \forall z \in \mathcal{M}$. Furthermore, assume that $f(c)<\infty$, where 
\begin{equation} \label{hd3}
f(u):= \sum_{r \ge 1}  r^{d(u-1)-1}\alpha(r)^{\frac{\chi}{\chi+u}}<\infty,
\end{equation}
and $c$ is the smallest even integer $c \ge \phi$. Then there exists a constant $C_1 $ 
$$
\E \Big| \sum_{z \in \mathcal{M}} \Xi_z \Big|^\phi \le C_1 \Big( \sum_{z \in \mathcal{M}}  \big\{ \E\big[| \Xi_z |^{\phi+\chi}\big]^{\frac{\phi}{\phi+\chi}}\big\} \lor 1 \Big)^{(\phi/2)\lor 1}.
$$
If $\phi = 2$, both '' $\lor 1$'' can be dropped in the above formula. 
The constant $C_1$ only depends on the mixing coefficients and can be upper bounded as follows: $C_1 \le C_2 C_3$, where $C_2$ only depends on $\phi$ and $\chi$ and $C_3 = C_3(f(1), f(2),...,f(c))$ is monotone in each component.
\end{prop}

\end{document}